\documentclass[conference]{IEEEtran}
\IEEEoverridecommandlockouts 
\usepackage{enumerate, enumitem}
\usepackage{graphicx, float}
\usepackage{amsmath, amssymb, amsthm, mathtools, stmaryrd}
\usepackage{cleveref, xcolor, todonotes}
\usepackage{array, tabularx, booktabs}
\usepackage{mdframed}

\def\BibTeX{{\rm B\kern-.05em{\sc i\kern-.025em b}\kern-.08em
    T\kern-.1667em\lower.7ex\hbox{E}\kern-.125emX}}

\newcommand{\lp}{\left(}
\newcommand{\rp}{\right)}
\newcommand{\lsb}{\left[}
\newcommand{\rsb}{\right]}

\newcommand{\mb}{\,\middle|\,}

\let\oldforall\forall
\renewcommand{\forall}{\oldforall\;}
\newcommand{\pad}{\addvspace{1.5ex}}

\newcommand{\Sp}{\mathcal S}

\newcommand{\proofover}{\hfill\IEEEQED\par\addvspace{1.5ex}}

\newcommand{\E}{\mathop{\mathbb{E}}\nolimits} 
\newcommand{\Prob}{\mathbb{P}} 
\newcommand{\R}{\mathbb{R}} 
\newcommand{\N}{\mathbb{N}} 

\theoremstyle{remark}
\newtheorem{remark}{Remark}

\newmdtheoremenv[
  linewidth=0pt,
  topline=true,
  bottomline=true,
  leftline=true,
  rightline=true,
  linecolor=black,
  backgroundcolor=white, 
  innertopmargin=6pt,
  innerbottommargin=6pt,
  innerleftmargin=6pt,
  innerrightmargin=6pt,
  skipabove=10pt,
  skipbelow=10pt,
]{lemma}{\upshape{\textbf{Lemma}}}

\newmdtheoremenv[
  linewidth=0pt,
  topline=true,
  bottomline=true,
  leftline=true,
  rightline=true,
  linecolor=black,
  backgroundcolor=white, 
  innertopmargin=6pt,
  innerbottommargin=6pt,
  innerleftmargin=6pt,
  innerrightmargin=6pt,
  skipabove=10pt,
  skipbelow=10pt,
]{claim}{\upshape{\textbf{Claim}}}

\newmdtheoremenv[
  linewidth=0.8pt,
  topline=true,
  bottomline=true,
  leftline=true,
  rightline=true,
  linecolor=black,
  backgroundcolor=white, 
  innertopmargin=6pt,
  innerbottommargin=6pt,
  innerleftmargin=6pt,
  innerrightmargin=6pt,
  skipabove=10pt,
  skipbelow=10pt,
]{proposition}{\upshape{\textbf{Proposition}}}

\newmdtheoremenv[
  linewidth=0.8pt,
  topline=true,
  bottomline=true,
  leftline=true,
  rightline=true,
  linecolor=black,
  backgroundcolor=white, 
  innertopmargin=6pt,
  innerbottommargin=6pt,
  innerleftmargin=6pt,
  innerrightmargin=6pt,
  skipabove=10pt,
  skipbelow=10pt,
]{corollary}{\upshape{\textbf{Corollary}}}

\newmdtheoremenv[
  linewidth=1pt,
  topline=true,
  bottomline=true,
  leftline=true,
  rightline=true,
  linecolor=black,
  backgroundcolor=white, 
  innertopmargin=6pt,
  innerbottommargin=6pt,
  innerleftmargin=6pt,
  innerrightmargin=6pt,
  skipabove=10pt,
  skipbelow=10pt,
]{theorem}{\upshape{\textbf{Theorem}}}

\mdtheorem[
  linecolor=black,
  linewidth=0.8pt,
  topline=true,
  bottomline=true,
  leftline=true,
  rightline=true,
  backgroundcolor=white,
  innertopmargin=6pt,
  innerbottommargin=6pt,
  innerleftmargin=4pt,
  innerrightmargin=5pt,
  skipabove=10pt,
  skipbelow=10pt,
  leftmargin=0pt,
  rightmargin=0pt,
  frametitlefont=\normalfont\bfseries,
  frametitlealignment=\raggedright,
  frametitleaboveskip=3pt,       
  frametitlebelowskip=-1pt,      
  font=\normalfont,
]{definition}{Definition}

\mdtheorem[
  linecolor=black,
  linewidth=0.8pt,
  topline=true,
  bottomline=true,
  leftline=true,
  rightline=true,
  backgroundcolor=white,
  innertopmargin=6pt,
  innerbottommargin=6pt,
  innerleftmargin=4pt,
  innerrightmargin=5pt,
  skipabove=10pt,
  skipbelow=10pt,
  leftmargin=0pt,
  rightmargin=0pt,
  frametitlefont=\normalfont\bfseries,
  frametitlealignment=\raggedright,
  frametitleaboveskip=3pt,       
  frametitlebelowskip=-1pt,      
  font=\normalfont,
]{conjecture}{Conjecture}

\begin{document}
\title{Optimal Service Mode Assignment in a Simple Computation Offloading System: Extended Version
\thanks{This work was supported by the National Science Foundation (NSF) under Grant Nos. CNS-2148183 and CNS-2148128.}}

\author{
\IEEEauthorblockN{Darin Jeff and Eytan Modiano}
\IEEEauthorblockA{LIDS, Massachusetts Institute of Technology, Cambridge, MA, USA\\
Email: \{djeff, modiano\}@mit.edu}
}

\maketitle

\begin{abstract}
We consider a simple computation offloading model where jobs can either be fully processed in the cloud or be partially processed at a local server before being sent to the cloud to complete processing.  Our goal is to design a policy for assigning jobs to service modes, i.e., full offloading or partial offloading, based on the state of the system, in order to minimize delay in the system.  We show that when the cloud server is idle, the optimal policy is to assign the next job in the system queue to the cloud for processing.  However, when the cloud server is busy, we show that, under mild assumptions, the optimal policy is of a threshold type, that sends the next job in the system queue to the local server if the queue exceeds a certain threshold. Finally, we demonstrate this policy structure through simulations.
\end{abstract}

\section{Introduction}\label{sec:intro}
\subsection{Motivation}
Modern computational workflows are increasingly resource-intensive and reliant on cloud computing infrastructure. This trend has accelerated significantly with the widespread adoption of Large Language Models (LLMs), such as ChatGPT and Grok, whose memory and compute intensive architectures typically surpass the capabilities of individual user devices.
Contemporary systems often rely almost exclusively on cloud-based processing, with little to no local computations. At the same time, processing capacities at intermediate network nodes, such as user devices, edge servers, and intermediate wireless nodes continue to improve. Leveraging these intermediate resources can enable data compression, reducing demands on communication links and enhancing the computational capacity of the system. Such hybrid strategies open opportunities to enhance overall system capacity and reduce end-to-end delay.

Figure~\ref{fig:motivation}  illustrates a representative service pipeline commonly encountered in modern offloading workflows. A recurring theme in such pipelines is that increased processing at an upstream stage reduces the workload at subsequent stages. In general, incoming jobs can be processed under multiple service modes, each placing different levels of load on the various stages. When one stage becomes relatively congested, a backpressure-driven intuition suggests that selecting service modes which shift load away from that stage can facilitate implicit load balancing. However, this introduces a fundamental delay trade-off: assigning more processing to an upstream stage can reduce downstream load and delay at this stage, but will increase the load and delay at upstream stages.

\begin{figure}[t]
\centering
\includegraphics[width=0.48\textwidth]{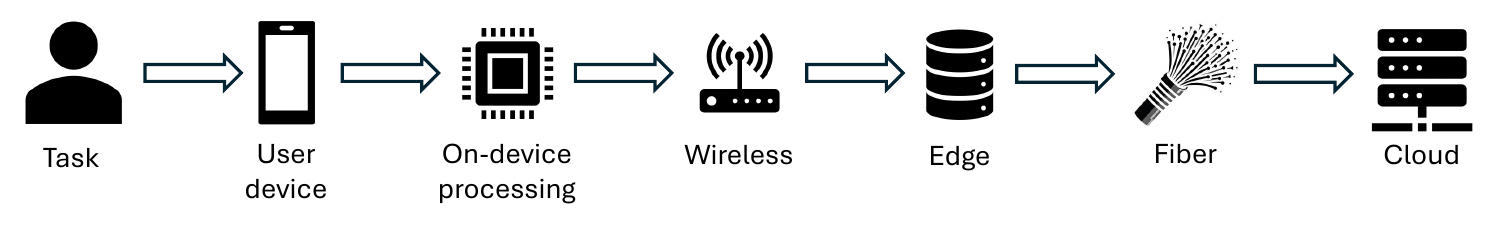}
\caption{Illustration of a typical offloading pipeline with multiple computationally capable network nodes.}
\label{fig:motivation}
\end{figure}

\subsection{Contributions}
We formally analyze this trade-off by modeling a two-stage system with two distinct service modes: (i) pure offloading, where the job is entirely processed at the second stage, and (ii) mixed service, where the job undergoes partial processing at the first stage before being forwarded. The system is modeled as a continuous-time Markov Decision Process (CTMDP) with the objective of minimizing average job delay (i.e. sojourn time). To characterize delay-optimal scheduling policies, we formulate the problem within a dynamic programming framework.

Our objective is to show that the optimal policy admits a simple structural characterization. We show that when the second-stage (cloud) resources are idle, it is always optimal to offload the next job. When the cloud is busy, under mild conditions, the optimal decision follows a \textit{switch-type} structure: assign the job to the mixed-processing mode if queue lengths exceed certain thresholds, and idle otherwise. Our analysis leverages a coupling-based approach to establish this structure, and we verify our results through simulation.

\subsection{Related work}
Trade-offs in processing networks have been studied from various perspectives. In~\cite{LucaLucaLuca}, the authors investigate optimal estimation under computation-communication trade-offs. The work in~\cite{LucaVish} considers the co-design problem for computation and communication, where multiple agents transmit updates to a common base station with the goal of minimizing Age of Information (AoI). In the context of service distribution,~\cite{Hao} proposes an extension of the Lyapunov drift-plus-penalty method that achieves throughput optimality while minimizing link-related costs.

Closer to our setting, scheduling in distributed computing networks with both packet transmission and processing requirements is addressed in~\cite{Jianan}. This work also focuses on throughput-optimal dynamic control policies, though delay minimization is not explicitly addressed.

In the limiting regime of our model -- where one service mode does all the processing locally and the other service mode offloads all the processing to the cloud, our formulation reduces to the classical heterogeneous server problem initially proposed in~\cite{Larsen}. Here, the delay-optimal policy was conjectured to have a threshold structure - a result later proved using policy iteration techniques in~\cite{LinK}, and subsequently simplified via coupling arguments in~\cite{Walrand}. While threshold-type policies are known to be optimal for two heterogeneous servers, extending this result to larger systems has remained an open problem. Several attempts have been made in this direction~\cite{Rykov,Luh}, but later scrutiny revealed issues in their proofs~\cite{Francis,Koole}.

\medskip
\noindent\textbf{Paper Structure.}
The remainder of the paper is organized as follows. Section~\ref{sec:sys_model} introduces the system model. In Section~\ref{sec:analysys}, we analyze the optimal policy and establish its structural properties. Section~\ref{sec:simulations} presents simulation results that verify our theoretical findings and compare the delay performance of the optimal policy with baseline strategies.  Section~\ref{sec:conclusion} concludes with open directions.

\section{System Model}\label{sec:sys_model}
\subsection{Operational Regime}
We consider an offloading system comprising two sequential servers: a local server with processing capacity $\mu_0$, and a cloud server with capacity $K\mu_0$, where $K>1$. This setup is illustrated in Figure~\ref{fig:sys_model}. Jobs arrive according to a Poisson process with rate $\lambda$. The system employs two representative service strategies:
\begin{itemize}
\item \textbf{Service Mode 1 (Offloading):} Jobs are sent directly to the cloud without local processing. The service time at the cloud under this mode is exponential with rate $\mu_{c1} = K\mu_0$.
\item \textbf{Service Mode 2 (Mixed):} Jobs receive a fraction $f$ of their total service requirement at the local server, and the remainder at the cloud. The service times at the local and cloud servers under this mode are independent exponential random variables with rates $\mu_{l2} = \mu_0/f$ and $\mu_{c2} = K\mu_0/(1-f)$, respectively.
\end{itemize}

Jobs first enter a centralized \textit{base queue}, as shown in Figure~\ref{fig:sys_model}, and remain there until assigned a service mode by a \textit{Dynamic Mode Selector} (DMS). The control decisions of mode assignment are taken at \textit{transition instants} — namely, upon a job arrival or a service completion at either server. Possible control actions include idling, assigning a job from the base queue to Service Mode 1 (SM1), or assigning it to Service Mode 2 (SM2).
The servers are non-idling and maintain separate queues for each service modes. The cloud uses Shortest Expected Processing Time (SEPT) scheduling, prioritizing SM2 over SM1 jobs.

Once assigned a service mode, a job must complete its service under that mode. To avoid premature commitment without compromising optimality, the DMS assigns jobs to SM1 (or SM2) only if no other SM1 (or SM2) jobs are queued at the cloud (or local server, respectively). Consequently, the local SM2 queue and cloud SM1 queue each require only unit storage capacity.

To enable effective load balancing between the servers, the fraction of processing assigned locally in SM2 must satisfy:
\[f > \frac{1}{K+1}\]
Given the server capacities $\mu_0$ and $K\mu_0$, the term $1/(K+1)$ represents the capacity of the local server relative to the total system capacity. Since SM1 jobs are processed entirely at the cloud, this condition enables the system to redistribute the load to the local server by assigning SM2.

To simplify our analysis, we introduce a slightly stronger assumption:
\[f > \frac{1}{K} \label{eqn:dominance}\]
which is equivalent to requiring that $\mu_{c1} > \mu_{l2}$. Since cloud resources typically dominate ($K \gg 1$), this condition only marginally tightens the requirement for load balancing and facilitates performance analysis without compromising the model's generality. 

\begin{figure*}[t]
	\centering
	\includegraphics[width=0.9\textwidth]{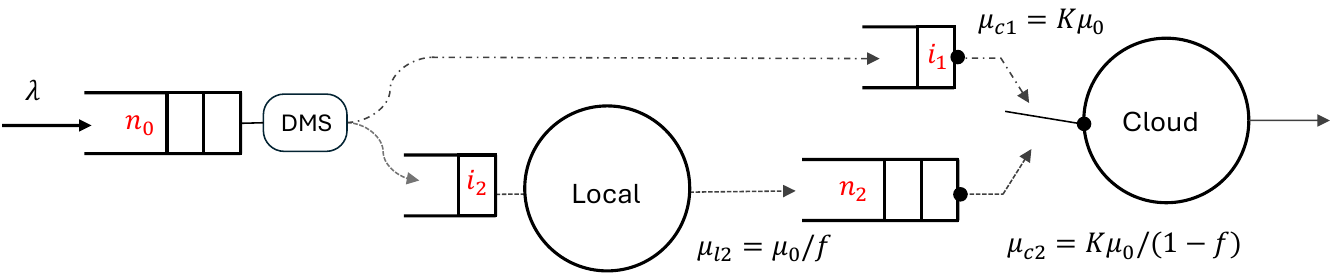}
	\caption{System model with two service modes: SM1 (pure offloading) and SM2 (mixed processing). Jobs are assigned to modes by the scheduler (DMS) based on system state. The cloud server maintains separate queues and prioritizes SM2 jobs using SEPT scheduling.}
	\label{fig:sys_model}
\end{figure*}

\begin{remark} \textbf{(Heterogenous case)}
While the system description above focuses on computational tasks, where job processing is split across the two stages in fractions $f$ and $1-f$, our model naturally extends to heterogeneous settings. In particular, the first and second stages may represent different types of operations, such as computation and communication, respectively, as in the on-device \textendash~wireless segment of the offloading pipeline shown in Figure~\ref{fig:motivation}.
In this generalized case, jobs under SM1 are served at rate $\mu_{c1}$ at the second stage, while SM2 jobs receive service at rates $\mu_{l2}$ and $\mu_{c2}$ at the first and second stages, respectively, with $\mu_{c2}>\mu_{c1}$. This system can be cast in our modeling framework with parameters $\mu_0 = \frac{\mu_{l2}(\mu_{c2}-\mu_{c1})}{\mu_{c2}}$, $K = \frac{\mu_{c1}\mu_{c2}}{\mu_{l2}(\mu_{c2}-\mu_{c1})}$ and $f = \frac{\mu_{c2}-\mu_{c1}}{\mu_{c2}}$. 
\end{remark}


\subsection{State and admissible controls}
To facilitate the analysis of this system using a dynamic programming framework, we introduce notation to rigorously describe the system state, and controls. Let $\mathbb{N}_0$ denote the set of non-negative integers. We define the system state as a 4-tuple
\[s \coloneqq (n_0, i_2, i_1, n_2) \in \Sp,\]
with state space $\mathcal{S} = \mathbb{N}_0 \times \{0,1\}^2 \times \mathbb{N}_0$. 

The interpretation of each component is summarized in Table~\ref{tab:statevars}.
We use 
\[n(s)\coloneqq n_0+i_2+i_1+n_2\]
to denote the total number of jobs in the system in state $s$.
\begin{table}[t]
\caption{Interpretation of State Variables}
\label{tab:statevars}
\centering
\begin{tabular}{|c|l|}
\hline
\textbf{Symbol} & \textbf{Description} \\
\hline
$n_0$ & Number of jobs in the base queue \\
$i_2$ & Indicator for a local SM2 job in service \\
$i_1$ & Indicator for a cloud SM1 job (in queue or service) \\
$n_2$ & Number of SM2 jobs at the cloud \\
\hline
\end{tabular}
\end{table}

Given a state $s$, admissible control actions at the DMS are:
\[
\mathcal{U}(s) = 
\begin{cases} 
\{\mathtt{idle,\!~SM1,\!~SM2}\}, 
  & \text{if } n_0 \!\geq\! 1,\ i_2 = i_1 = 0 \\
\{\mathtt{idle,\!~SM1}\}, 
  & \text{if } n_0 \!\geq\! 1,\ i_2 = 1,\ i_1 = 0 \\
\{\mathtt{idle,\!~SM2}\}, 
  & \text{if } n_0 \!\geq\! 1,\ i_2 = 0,\ i_1 = 1 \\
\{\mathtt{idle}\}, 
  & \text{otherwise}
\end{cases}
\]

with each control action defined as follows:
\begin{description}
    \item \texttt{idle}: No service mode assignment is made.
    \item \texttt{SM1}: A job from the base queue is assigned to the SM1 cloud queue.
    \item \texttt{SM2}: A job from the base queue is assigned to the SM2 local queue.
\end{description}

\noindent
\textbf{Note:} We use '\texttt{SM1}' to refer to the control action selected by the DMS, and plain 'SM1' to denote the corresponding mode or queue.

To concisely describe state transitions due to arrivals, service completions, and service mode assignments, we define state transition and service-mode operators in Table~\ref{tab:state_transitions}. These operators map a state $s\in\Sp$ to another state $s'\in\Sp$. For example, if a system is in state $s_0\in\Sp$ and experiences an arrival, transitioning to state $s_1 = s_0 + (1,0,0,0)$, and subsequently assigns a job using \texttt{SM1}, transitioning to $s_2 = s_1 + (-1,0,1,0)$, we can concisely represent this as $s_1 = As_0$ and $s_2 = U_1As_0$.

\begin{table*}[t]
\centering
\caption{State Transition Operators}
\label{tab:state_transitions}
\begin{tabular}{l l l l}
\toprule
\textbf{Operator} & \textbf{Definition} & \textbf{Domain} & \textbf{Description} \\
\midrule
$A$  & $(n_0, i_2, i_1, n_2) \mapsto (n_0 + 1, i_2, i_1, n_2)$ 
& $\Sp$  & Arrival \\
$D_1$  & $(n_0, i_2, i_1, n_2) \mapsto (n_0, i_2, (i_1 - 1)^+, n_2)$ 
& $i_1 = 1$ & Cloud SM1 departure \\
$D_2$ & $(n_0, i_2, i_1, n_2) \mapsto (n_0, i_2, i_1, (n_2 - 1)^+)$ 
& $n_2 \in \N$ & Cloud SM2 departure \\
$D_L$ & $(n_0, i_2, i_1, n_2) \mapsto (n_0, (i_2 - 1)^+, i_1, n_2 + i_2)$ 
& $i_2 = 1$ & Local SM2 completion \\
$U_1$ & $(n_0, i_2, i_1, n_2) \mapsto (n_0 - 1, i_2, i_1 + 1, n_2)$ 
& $n_0 \in \N,\ i_1 = 0$ & Assign to SM1 \\
$U_2$ & $(n_0, i_2, i_1, n_2) \mapsto (n_0 - 1, i_2 + 1, i_1, n_2)$ 
& $n_0 \in \N,\ i_2 = 0$ & Assign to SM2 \\
\bottomrule
\end{tabular}
\end{table*}


\subsection{Dynamic Programming Formulation}
To quantitatively compare delay performance under various system policies, we utilize a continuous-time dynamic programming (DP) framework. The DP formulation involves the state space $\mathcal S$, the state-dependent action space $\mathcal U(s)$, cost rate function $c(s) = n(s)$ and discount factor $\beta$.

A \textit{system policy}, $\pi$ prescribes state-dependent control actions for the Dynamic Mode Selector (DMS). Formally, given the state $s \in \mathcal{S}$, the policy specifies an action:
\[\pi(s) \in \mathcal{U}(s)\]
We focus exclusively on stationary Markovian policies since they are known from standard DP theory (e.g., \cite{Bertsekas}), to be optimal for minimizing the expected discounted cost.\\

\noindent
\textbf{Note:} If a policy wants to assign two jobs from the base queue one to SM1 and one to SM2, it can do this by picking the respective control actions in immediate succession. Without loss of generality, we require control action \texttt{SM1} to be \textbf{picked first} in such cases.\\

Given an initial state $s_0\in\Sp$ and a policy $\pi$, the \textit{discounted cost-to-go} is defined as:
\[J_\beta(s_0; \pi) = \mathbb{E}_\pi\left[\int_0^\infty e^{-\beta t} n(t) \, dt \mid s(0) = s_0\right]\]
and the corresponding \textit{value function} at state $s_0$ is:
\[V_\beta(s_0) = \min_{\pi} J_\beta(s_0; \pi).\]

Finally, a policy $\pi^*$ is termed \textit{delay optimal} if it minimizes the expected discounted delay cost among all stabilizing policies, within the stability region, i.e.
\[J_\beta(s_0; \pi^*) = V_\beta(s_0) \quad \text{for all } s_0 \in \mathcal{S}.\]


\subsection{Special Policies}
In our main result, we identify specific structural properties of the optimal policy. To formalize these properties, we introduce two special classes of policies.

First, define
\[\mathcal{S}_1 \coloneqq \{s = (n_0, i_2, i_1, n_2) \in \mathcal{S} \mid n_0 \in \mathbb{N}, \, i_1 = n_2 = 0\},\]
representing states where the cloud queues are empty and jobs are available in the base queue. Since the cloud prioritizes SM2 jobs, assigning a job to SM1 in a state in $\Sp\setminus\Sp_1$ (where the SM2 cloud queue is non-empty), leads to that job being blocked until all SM2 jobs are processed. Hence an optimal policy doesn't assign SM1 in states outside $\Sp_1$.
Knowing that SM1 is not assigned outside $\Sp_1$, we now define a policy to be \textit{cloud-first} if it assigns SM1 for all states in $\Sp_1$.

\begin{definition}[{\normalfont\itshape Cloud-First Policy}] \label{def:cloud_first}
A policy $\pi$ is said to be \textit{cloud-first} if
\[\pi(s) = \mathtt{SM1} \quad \forall s \in \Sp_1.\]
where,
\[\mathcal{S}_1 \coloneqq \{s = (n_0, i_2, i_1, n_2) \in \mathcal{S} \mid n_0 \in \mathbb{N}, \, i_1 = n_2 = 0\},\]
\end{definition}

Identifying an optimal policy as cloud-first fully characterizes its SM1 assignments, reducing the remaining specification to states in $\Sp\setminus\Sp_1$, where possible actions are only either assigning SM2 or idling.

We now define a class of policies that have a switching structure in their use of control \texttt{SM2}.

\begin{definition}[{\normalfont\itshape Switch-Type Policy}] \label{def:switch-type}
A policy $\pi$ is said to be \textit{switch-type} if it satisfies the following monotonicity property:
\[\pi(s) = \mathtt{SM2} \Rightarrow \pi(s + (x_0, 0, 0, x_2)) = \mathtt{SM2}\]
for all $s \in \Sp$ and $(x_0, x_2) \in \mathbb N_0^2$.
\end{definition}

This property means that if a policy assigns SM2 in a given state, it also assigns SM2 in states derived by adding jobs to either the base queue or the SM2 queue.


\section{Analysis}\label{sec:analysys}
\subsection{Coupling Framework}\label{sec:coupling_framework}
Before stating our main result, we introduce a path-wise coupling argument to compare system trajectories under different policies. Similar arguments have been used in \cite{Walrand} to show the delay optimality of threshold policies in the M/M/2 system with heterogenous servers.

To show that a policy $\pi$ is suboptimal in this framework, we consider a system, $S$ operating under $\pi$ with system state $S(t)$. We then consider a modified policy, $\tilde\pi$ governing a coupled system, $\tilde S$ with state process $\tilde S(t)$. We use $n(t), \tilde n(t)$ to denote the number of jobs in $S,\tilde S$ respectively. Each job $J$ arriving into either system is associated with a triplet of service times: $(\sigma_{c1}^J, \sigma_{l2}^J, \sigma_{c2}^J)$, where the elements denote:
\begin{itemize}
    \item $\sigma_{c1}^J$: Service time for $J$ at the cloud under SM1,
    \item $\sigma_{l2}^J$: Service time for $J$ at the local server under SM2,
    \item $\sigma_{c2}^J$: Service time for $J$ at the cloud under SM2.
\end{itemize}

Since only one of the service modes ($\sigma_{c1}^J$ or $(\sigma_{l2}^J, \sigma_{c2}^J)$) is realized along any given sample path of a system, we couple $\sigma^J_{l2}$ and $\sigma_{c1}^J$ by setting:
\begin{equation}\sigma^J_{l2} = \frac{\mu_{c1}}{\mu_{l2}} \sigma^J_{c1}\label{eqn:job_self}\end{equation}
This coupling is valid due to properties of exponential random variables, allowing simple scaling to match distributions.
Finally, we couple the systems to have the same arrival process such that each job $J$ arriving into $S$ and the corresponding coupled job, $\tilde J$ arriving into $\tilde S$ have the same triplet of service times.
We remark that under this coupling:
\begin{itemize}
	\item Coupled jobs assigned the same service mode in both systems will have the same service time
	\item For coupled jobs assigned SM1 in one system and SM2 in the other, the processing time for the job assigned SM2 is strictly greater than the coupled job assigned SM1 from \eqref{eqn:job_self} and since $\mu_{c1}>\mu_{l2}$.
\end{itemize}
The sub-optimality of $\pi$ is established by showing
\[\tilde n(t) \leq n(t) \quad \forall t\in\R^+\]
with the inequality being strict over a non-empty interval with positive probability.

\subsection{Main Result}\label{sec:main_result}
Here, we state the main result of our analysis.
\begin{theorem}
The delay-optimal policy is \textit{cloud-first}. Furthermore, assuming the \textit{urgency-monotonicity} conjecture, the policy is \textit{switch-type}.
\end{theorem}
\begin{IEEEproof}
This result is established through a sequence of supporting theorems which we prove subsequently. First, Theorem~\ref{th:cloud_first} shows that the delay-optimal policy is \textit{cloud-first}. Then, Theorem~\ref{th:partial_switch} demonstrates that any optimal \textit{cloud-first} policy has a partial switch-type structure. Finally, assuming the \textit{urgency-monotonicity conjecture} (stated below), Theorem~\ref{th:switch_type} concludes that the optimal policy is fully \textit{switch-type}.
\end{IEEEproof}

\pad
Before delving into the technical details of the proofs, we provide intuition motivating these structural results.

The \textit{cloud-first} property is motivated by the assumption that the cloud has strictly higher capacity than the local server. As a result, it is never optimal to leave the cloud idle when there are jobs available in the base queue that can be offloaded. This implies the policy should always prioritize offloading to the cloud whenever it is idle.

To gain intuition for the emergence of a switch-type structure, consider the comparative analysis between service modes SM1 and SM2. Under our operational assumptions, we have $\mu_{c1} > \mu_{l2}$, implying that the total expected processing time for SM2 ($1/\mu_{l2} + 1/\mu_{c2}$) is strictly greater than for SM1 ($1/\mu_{c1}$). Based solely on processing times, SM1 appears more efficient.

However, in scenarios where the cloud is backlogged, initiating processing locally under SM2 can reduce the job's subsequent time spent in the cloud queue. Despite SM2's higher nominal processing time, it can effectively bypass some cloud-side processing delay by front-loading processing at the local server. Consequently, once assigning SM2 becomes optimal at a certain state, it remains optimal as backlog levels increase, leading to a natural switch-type policy structure.

\subsection{Structural Conjecture}
Although we do not provide a rigorous proof for this conjecture, it is consistent with structural insights from related dynamic queueing models in the existing literature (e.g.,~\cite{LinK, Hajek})
and supported by simulation evidence in our setting. A similar conjecture is employed in the related open problem of delay minimization in the M/M/3 system with heterogenous servers, where proving explicit threshold structures remains analytically challenging \cite{Koole}. 

Even without this conjecture, we can still establish \textit{cloud-first} behavior and a partial switch-type structure in the optimal policy. However, achieving a complete \textit{switch-type} characterization requires the conjecture.

\begin{conjecture*}[\normalfont\itshape Urgency-Monotonicity] 
Let $\pi^*$ be a delay-optimal policy. Then for all $s\in\mathcal S$,
\begin{align}
	\pi^*(As) = \mathtt{idle} \Rightarrow \pi^*(s) = \mathtt{idle}
\end{align} 
\end{conjecture*}

This conjecture posits that if idling is optimal in a state with more jobs, then reducing the number of jobs in the base queue does not make non-idling preferable. Intuitively, fewer jobs in the base queue reduce urgency, reinforcing the optimality of an idling action.\\
Recall that a switch-type policy exhibits switching behavior in both increasing directions of the base queue and the SM2 cloud queue when assigning SM2. The conjecture, however, only imposes monotonicity in the direction of the base queue for the action 	\texttt{idle}, offering a natural and minimal structural condition.

\subsection{Cloud-First Behavior at DMS}\label{sec:dms_policy}

In this section, we establish that an optimal policy must be \textit{cloud-first}.

\begin{theorem} \label{th:cloud_first} (Cloud-First Behaviour)
Let $\pi^*$ be a delay-optimal policy, then $\pi^*$ is \textit{cloud-first} (Def~\ref{def:cloud_first}).
\end{theorem}
\begin{IEEEproof}
This result follows directly from Propositions~\ref{pro:offload_empty} and~\ref{pro:SM1_full}, which we state and prove subsequently. From these propositions, we obtain:
\[\pi^*(n,0,0,0) = \mathtt{SM1}, \quad \forall n\in\mathbb N.\]
and
\[\pi^*(n,1,0,0) = \mathtt{SM1}, \quad \forall n\in\mathbb N.\]
We combine these results, we have:
\[\pi^*(s) = \mathtt{SM1}, \quad \forall s\in\Sp_1\]
where,
\[\mathcal{S}_1 \coloneqq \{s = (n_0, i_2, i_1, n_2) \in \mathcal{S} \mid n_0 \in \mathbb{N}, \, i_1 = n_2 = 0\}.\]
\end{IEEEproof}
\begin{proposition} \label{pro:offload_empty}
Let $\pi^*$ be a delay-optimal policy. Then: 
\[\pi^*(n,0,0,0) = \mathtt{SM1}, \quad \forall n\in\mathbb N.\]
\end{proposition} 
\begin{IEEEproof}
This follows directly from Lemmas~\ref{lem:no_idle_empty} and~\ref{lem:no_SM2_empty}.
\end{IEEEproof}

\begin{lemma} \label{lem:no_idle_empty}
Let $\pi^*$ be a delay-optimal policy. Then: 
\[\pi^*(n,0,0,0) \neq \mathtt{idle}, \quad \forall n\in\mathbb N.\]
\end{lemma}
\begin{proof}
See Appendix~\ref{app:no_idle_empty}.
\end{proof}

\begin{lemma} \label{lem:no_SM2_empty}
Let $\pi^*$ be a delay-optimal policy. Then: 
\[\pi^*(n,0,0,0) \neq \mathtt{SM2}, \quad \forall n\in\mathbb N.\]
\end{lemma}
\begin{proof}
See Appendix~\ref{app:no_SM2_empty}.
\end{proof}
\pad
Proposition~\ref{pro:offload_empty} states that assigning SM1 is optimal whenever the base queue contains jobs while other queues are empty. Doing this transitions the system to a state $(n_0-1,0,1,0)$, for some $n_0\in\N_0$, where the system can now either idle or assign SM2. We now highlight in Proposition~\ref{pro:empty_behaviour} that the optimal policy at these states exhibits a threshold-based switching structure, shifting from idling to assigning SM2 as $n_0$ increases.

\begin{proposition} \label{pro:empty_behaviour}
Let $\pi^*$ be a delay-optimal policy. Then, there exists some threshold $N_0\in \mathbb \N_0$ such that:
\begin{equation}
\pi^*(n_0,0,1,0) = 
\begin{cases}
\mathtt{idle} & \text{if } n_0 < N_0\\
\mathtt{SM2} & \text{if } n_0 \geq N_0
\end{cases} \label{eqn:empty_policy}
\end{equation}
\end{proposition}
\begin{IEEEproof}
The proof relies on Lemmas~\ref{lem:empty_interval} and~\ref{lem:eventually_SM2}, proven subsequently. Lemma~\ref{lem:eventually_SM2} ensures that for each $n \in \mathbb{N}$, there exists $m_n > n$ satisfying $\pi^*(m_n,0,1,0) = \mathtt{SM2}$. Define:\[
N_0 \coloneqq \min\{n\in\mathbb N \mid \pi(n,0,1,0) =\mathtt{SM2}\}.
\]
By Lemma~\ref{lem:empty_interval}, we must have $\pi(n,0,1,0) = \mathtt{idle}$ for all $n < N_0$ and $\pi(n,0,1,0) = \mathtt{SM2}$ for all $n \geq N_0$. This completes the proof.
\end{IEEEproof}
\begin{lemma} \label{lem:empty_interval}
Let $\pi^*$ be a delay-optimal policy.
If, for some $m_1,m_2\in\N_0$ with $m_1<m_2$:
\[\pi^*(m_1,0,1,0) = \pi^*(m_2,0,1,0) = \mathtt{SM2}.\]
then it follows that:
\[\pi^*(n,0,1,0) = \mathtt{SM2} \quad \forall n\in\llbracket m_1,m_2\rrbracket\]
\end{lemma}

\begin{proof}
See Appendix \ref{app:empty_interval}.
\end{proof}
\begin{lemma} \label{lem:par_comp_win}
Let $\pi^*$ be optimal and satisfy
\[\pi^*(n,0,1,0) =\mathtt{idle} \quad \forall n\geq M_0\]
for some $M_0\in \mathbb N$.
Then, the value function satisfies:
\[V_\beta(n_0,0,1,0)-V_\beta(n_0,0,0,1) \geq \epsilon_\beta \quad \forall  n_0 \geq M_0,\]
where $\epsilon_\beta>0$ is a constant independent of $n_0$.
\end{lemma}

\begin{proof}
See Appendix \ref{app:par_comp_win}.
\end{proof}
\pad
Lemma~\ref{lem:par_comp_win} states that if a job assigned SM2 completes local service, the system is in a consistently better state than if that job were still queued up and waiting to be served under SM1. This result helps us establish the inevitability of eventually assigning SM2 under an optimal policy. We formalize this in the following Lemma.

\begin{lemma} \label{lem:eventually_SM2}
Any policy $\pi$ that satisfies
\[
\pi(n,0,1,0) =\mathtt{idle} \quad \forall n\geq M_0
\]
for some $M_0\in \mathbb N$. Then, policy $\pi$ is not delay-optimal.
\end{lemma}
\begin{proof}
See Appendix \ref{app:eventually_SM2}
\end{proof}
\addvspace{1.5ex}
Having established a threshold structure in the optimal assignment of SM2 among states in $\{(n,0,1,0) \mid n\in\mathbb{N}\}$, we now restrict our attention exclusively to policies adhering to this structure. We explicitly define:
\begin{equation}
N_0 \coloneqq \min\{n\in\mathbb N \mid \pi^*(n,0,1,0) = \mathtt{SM2}\}. \label{eqn:N0def}
\end{equation}

Having previously shown that assigning SM1 is optimal at states in  $\{(n,0,0,0)\mid n_0\in\N\}$, it remains to show that SM1 is also optimal at states in  $\{(n,0,1,0)\mid n_0\in\N\}$. We show this in the following Proposition.
\begin{proposition}\label{pro:SM1_full}
Let $\pi^*$ be a delay-optimal policy, Then, 
\[\pi^*(n,1,0,0) = \mathtt{SM1}, \quad \forall n\in \mathbb N.\]
\end{proposition}
\begin{IEEEproof}
We first demonstrate in Lemma~\ref{lem:SM1_afterN_0} that:
\[\pi^*(n,1,0.0) = \mathtt{SM1} \quad \text{for } n\geq N_0-1\]
Proposition~\ref{pro:SM1_induction} then inductively extends this property downward to all $n\in \mathbb{N}$.
\end{IEEEproof}

\begin{lemma}\label{lem:SM1_afterN_0}
Let $\pi^*$ be a delay-optimal policy. Then
\[\pi^*(n,1,0,0) = \mathtt{SM1} \quad \text{for } n\geq N_0.\]
\end{lemma}
\begin{proof}
Let $s = (n+1,0,0,0)$ for some $n\geq N_0$.\\
From Propositions~\ref{pro:offload_empty} and~\ref{pro:empty_behaviour}, we have:
\[\pi^*(s) = \mathtt{SM1} \quad \text{and} \quad \pi^*(U_1s) = \mathtt{SM2}\]
indicating that the optimal policy successively assigns SM1 and then SM2 at state $s$. Since assigning both SM1 and SM2 is optimal at this higher state, assigning SM1 must also be optimal at the lower state, $U_2 s$.
\end{proof}


\begin{proposition} \label{pro:SM1_induction}
Let $\pi^*$ be a delay-optimal policy. Consider a state $s = (n_0,1,0,0)$ for some $n_0 \in \mathbb N$. Then:
\[\pi^*(A^2s) = \pi^*(As) = \mathtt{SM1} \Rightarrow \pi^*(s) = \mathtt{SM1}\]
\end{proposition}
\begin{IEEEproof}
Let $s = (n_0,1,0,0)$ for some $n_0\in\mathbb{N}$. At this state, the admissible DMS controls are $\mathcal{U}(s)=\{\mathtt{idle, SM1}\}$. Thus, it suffices to show that the scenario:
\[\pi^*(s) = \mathtt{idle},\quad \pi^*(A s) = \pi^*(A^2 s) = \mathtt{SM1}\]
cannot occur under optimality. Suppose for contradiction this scenario holds. Consider the state $D_Ls$, where the admissible controls are $\mathcal U(D_Ls) = \{\mathtt{idle, SM2}\}$. We examine two cases:

\textbf{Case 1:} $\pi^*(D_Ls) = \mathtt{SM2}$\\
This case is ruled out by Lemma~\ref{lem:idle_SM2}, proven subsequently.

\textbf{Case 2:} $\pi^*(D_Ls) = \mathtt{idle}$\\
This case is ruled out by Lemma~\ref{lem:idle_idle}, also proven subsequently.
Since both cases lead to contradictions, the original assertion follows.
\end{IEEEproof}

\begin{lemma}\label{lem:idle_SM2}
Let $s = (n_0, 1, 0, 0)$ for some $n_0\in \mathbb N$, and suppose there exists a policy $\pi$ such that:
\begin{align*}
	\pi(s) = \mathtt{idle},\quad
	\pi(As) = \mathtt{SM1},\quad
	\pi(D_Ls) = \mathtt{SM2}
\end{align*}
Then, $\pi$ is not delay-optimal.
\end{lemma}
\begin{proof}
	See Appendix~\ref{app:idle_SM2}.
\end{proof}

\begin{lemma}\label{lem:idle_idle}
Let $s = (n_0, 1, 0, 0)$ for some $n_0\in \mathbb N$, and suppose there exists a policy $\pi$ satisfying:
\begin{align*}
	\pi(s_0) = \:&\pi(D_Ls_0) = \mathtt{idle}\\
	\pi(As_0) = \:&\pi(A^2s_0) = \mathtt{SM1}
\end{align*}
Then, $\pi$ is not delay-optimal.
\end{lemma}
\begin{proof}
See Appendix~\ref{app:idle_idle}.
\end{proof}

\noindent
This completes the analysis underpinning Theorem~\ref{th:cloud_first}, establishing that the optimal policy must be \textit{cloud-first}. 


\subsection{Switch-Type Behavior at DMS}\label{sec:dms_switch}
Recall that identifying the optimal policy as \textit{cloud-first} completely characterizes SM1 assignments, leaving only the actions of assigning SM2 or idling to be determined at states in $\Sp\setminus\Sp_1$.
In Theorem~\ref{th:partial_switch}, we now establish that an optimal, \textit{cloud-first} policy exhibits a partial \textit{switch-type} structure in this region.

\begin{theorem}[Partial Switch-Type Structure]\label{th:partial_switch}
Let $\pi^*$ be a delay-optimal, cloud-first policy. Then there exist non-increasing sequences of integers $\{N_k\}_{k \in \N_0}$ and $\{N_k'\}_{k \in \N}$ defining a family of switching states:
\[\mathcal{S}^\mathrm{sw} \coloneqq \{(N_k, 0, 1, k), \; (N_k', 0, 0, k) \mid k \in \mathbb{N} \},\]
Such that for all $s\in\mathcal{S}^\mathrm{sw}$ and $(x_0,x_2)\in\N_0^2$
\[\pi^*\big(s + (x_0, 0, 0, x_2)\big) = \mathtt{SM2}\]
\end{theorem}

Theorem~\ref{th:partial_switch} highlights a near-complete \textit{switch-type} characterization. For each $k \in \mathbb N$, the states $(N_k, 0, 0, k)$ and $(N_k' - 1, 0, 1, k)$ serve as effective switching points beyond which assigning \texttt{SM2} becomes optimal in the increasing base-queue size or SM2 cloud queue size directions. For a complete \textit{switch-type} characterization, it remains to show that idling is optimal in the reverse directions, namely, as we decrease base-queue or SM2 cloud queue size. This final step employs the \textit{urgency-monotonicity} conjecture.

\begin{IEEEproof}[Proof of Theorem~\ref{th:partial_switch}]
The theorem follows directly from Lemma~\ref{lem:dec_switches}, establishing the sequences' monotonicity, and Proposition~\ref{pro:the_moving_wall}, showing optimal switching to SM2 in both increasing base-queue and SM2 cloud-queue directions. These results are stated and proved subsequently.
\end{IEEEproof}

\begin{proposition} \label{pro:wall_induction}
Let $\pi$ be a \textit{cloud-first}, delay-optimal policy.
If there exists a state $s \in \Sp$ such that:
\[\pi(s + (x_0,0,0,0)) = \mathtt{SM2} \quad \forall x_0\in \mathbb N_0\]
Then:
\[\pi(s + (0,0,0,1)) = \mathtt{SM2}\]
\end{proposition}
\begin{IEEEproof}
See Appendix~\ref{app:wall_induction}.
\end{IEEEproof}

Proposition~\ref{pro:wall_induction} is a key result supporting the construction of the partial switch-type structure.

\begin{proposition} \label{pro:the_wall}
Let $\pi^*$ be a delay-optimal policy, and define states:
\[s = (N_0, 0, 1, 0) \text{ and } s' = (N_0+1, 0, 0, 1)\]
with $N_0$ as defined in \eqref{eqn:N0def}.
Then:
\begin{enumerate}[label=(\Alph*)]
	\item \label{pro:parta}$\pi^*(s + (x_0, 0, 0, x_2)) = \mathtt{SM2} \quad \forall (x_0, x_2) \in \mathbb N_0^2$
	\item \label{pro:partb}$\pi^*(s' + (x_0, 0, 0, x_2)) = \mathtt{SM2} \quad \forall (x_0, x_2) \in \mathbb N_0^2$
\end{enumerate}
\end{proposition}

\begin{IEEEproof}
See Appendix~\ref{app:the_wall}.
\end{IEEEproof}

\addvspace{1.5ex} We now define the following thresholds for each value of the SM2 cloud queue length, $k\in\N$:
\begin{align}
	N_k &\coloneqq \min\left\{n \in \mathbb N \mid \pi^*(m,0,1,k) = \mathtt{SM2} \;\; \forall m \geq n \right\}, \label{eqn:Nkdef}\\
	N_k' &\coloneqq \min\left\{n \in \mathbb N \mid \pi^*(m,0,0,k) = \mathtt{SM2} \;\; \forall m \geq n \right\}. \label{eqn:Nk'def}
\end{align}
From Proposition~\ref{pro:the_wall}, these thresholds exist and are finite, as they are bounded above by $N_0$.

\begin{proposition}\label{pro:the_moving_wall}
Let $\pi^*$ be delay-optimal. Define:
\[s_k = (N_k, 0, 1, k), \quad s_k' = (N_k', 0, 0, k).\]
Then, for all $k \in \mathbb N$:
\begin{enumerate}[label=(\Alph*)]
    \item $\pi^*(s_k + (x_0, 0, 0, x_2)) = \mathtt{SM2} \quad \forall (x_0, x_2) \in \mathbb N_0^2$
    \item $\pi^*(s_k' + (x_0, 0, 0, x_2)) = \mathtt{SM2} \quad \forall (x_0, x_2) \in \mathbb N_0^2$
\end{enumerate}
\end{proposition}
\begin{IEEEproof}
This is a direct consequence of Proposition~\ref{pro:wall_induction}.
\end{IEEEproof}

\begin{lemma}\label{lem:dec_switches}
	The sequences $\{N_k\}_{k \in \N_0}$ and $\{N_k'\}_{k \in \N}$ defined above are non-increasing.
\end{lemma}
\begin{proof}
See Appendix~\ref{app:dec_switches}.
\end{proof}

\begin{theorem} \label{th:switch_type} (Switch-Type Structure)
Let $\pi^*$ be delay-optimal and cloud-first. Under the \textit{urgency-monotonicity} conjecture, $\pi^*$ is \textit{switch-type}.
\end{theorem}
\begin{IEEEproof}
From Proposition~\ref{pro:the_moving_wall}, we have for states defined as $s_k = (N_k, 0, 1, k)$:
\[\pi^*(s_k + (x_0, 0, 0, x_2)) = \mathtt{SM2} \quad \forall (x_0, x_2) \in \mathbb N_0^2\]
Note that by definition, $\pi^*(N_k-1,0,1,k)$ is always \texttt{idle}. Applying the \textit{urgency-monotonicity} conjecture, we obtain for each $k\in\N$:
\begin{equation}
\pi^*(n_0,0,1,k) = \begin{cases} \mathtt{SM2} & \text{if } n_0 \geq N_k \\ \mathtt{idle} & \text{if } n_0 < N_k\end{cases} \label{eqn:switching}
\end{equation}
Similarly, for $N_k'$ defined in \eqref{eqn:Nk'def} we get:
\begin{equation}
\pi^*(n_0,0,0,k) = \begin{cases} \mathtt{SM2} & \text{if } n_0 \geq N_k' \\ \mathtt{idle} & \text{if } n_0 < N_k'\end{cases} \label{eqn:switching'}
\end{equation}

\noindent Since the policy is \textit{cloud-first}, \eqref{eqn:switching} and \eqref{eqn:switching'} fully characterize $\pi$ in states where \texttt{SM2} is admissible.

\noindent Additionally, since $\{N_k\}_{k \in \N_0}$ and $\{N_k'\}_{k \in \N}$ are non-increasing, we have:
\[\pi^*(s) = \mathtt{SM2} \Rightarrow \pi^*(s + (x_0,0,0,x_2)) = \mathtt{SM2}\]
$\forall (x_0,x_2)\in\N_0^*$, asserting a \textit{switch-type} structure.
\end{IEEEproof}

\section{Simulations}\label{sec:simulations}
In this section, we empirically validate the structural properties of the delay-optimal policy established in our analysis. We also benchmark its performance against two baseline scheduling policies.
\subsection{Simulation Setup}
Uniformization is a classical technique in the theory of continuous-time Markov decision processes (CTMDPs), which transforms a CTMDP into an equivalent discrete-time MDP \cite{lippman}. We adopt this approach to convert our discounted continuous-time model into a discounted discrete-time formulation.

To compute the value function $V(s)$, we apply the standard value iteration algorithm. At each iteration, the Bellman equation is used to update the cost-to-go for every state. Upon convergence, 
we obtain the optimal policy by selecting the action that minimizes the updated value function at each state.

Our simulation framework uses the following settings: we impose a buffer limit of $n_{\max} = 300$ on both the base queue and the SM2 cloud queue. The discount factor for the discrete-time system is set to $\alpha = 1 - 10^{-5}$. 

We define the system load as $\rho = \lambda/(K+1)\mu_0$,
and normalize the capacity of the local server by fixing $\mu_0 = 1$. As a result, the tunable parameters in our simulation are:
$\rho$, which determines system utilization; $K$, which captures the relative capacity of the cloud; and $f$, which specifies the fraction of local processing in service mode 2.


\subsection{Cloud-First and Switch-Type structure}
We now empirically verify that the optimal policy exhibits both the \textit{cloud-first} and \textit{switch-type} structural properties by examining four system configurations, labeled (a) through (d).

Configuration (a) serves as the baseline, with parameters $\rho = 0.4$, $f = 0.4$, and $K = 8$. In configuration (b), we increase system utilization to $\rho = 0.8$, while keeping $f = 0.4$ and $K = 8$ fixed. Configuration (c) increases the degree of local processing in SM2 by setting $f = 0.8$ (with $\rho = 0.4$ and $K = 8$), and configuration (d) boosts cloud capacity with $K = 15$ (keeping $\rho = 0.4$ and $f = 0.4$).

We visualize the policy using 2D grids that represent slices of the full system state space. Recall that the system state is given by the quadruple $(n_0, i_2, i_1, n_2)$. In each grid, we fix two of these variables and use the remaining two as axes to represent the queue lengths. Each cell in the grid corresponds to a unique state and is annotated with the corresponding action: \texttt{0} for \texttt{idle}, \texttt{1} for \texttt{SM1}, and \texttt{2} for \texttt{SM2}. 

In Figure~\ref{fig:cloud_first}, we show the policy restricted to states in $\Sp_1 = \{(n_0,i_2,i_1,n_2)\in\Sp\mid i_1 = n_2 = 0\}$. The figure confirms that across all system configurations, the optimal policy always assigns SM1 whenever admissible, verifying the \textit{cloud-first} property established analytically.

\begin{figure}[!b]
\centering
\includegraphics[width=0.45\textwidth]{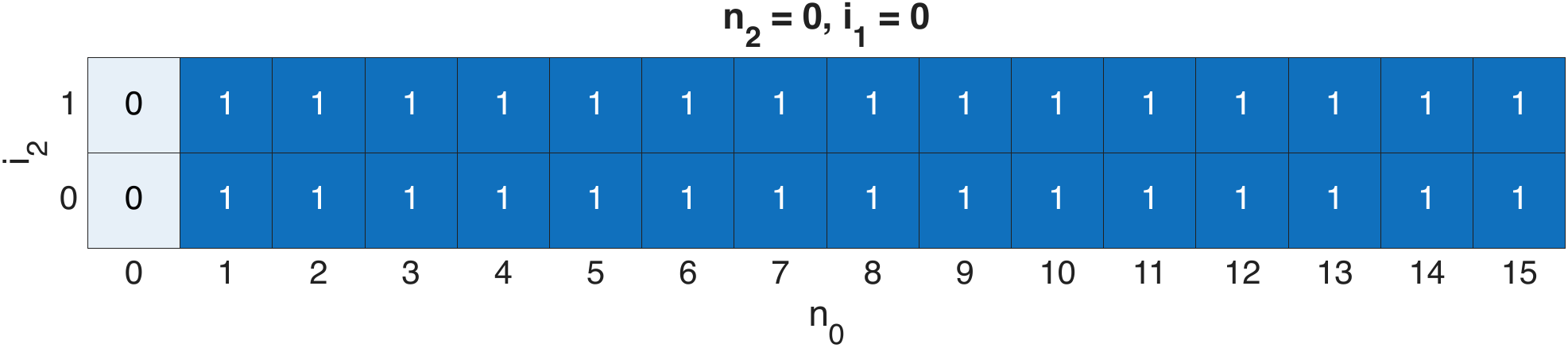}
\caption{\textit{Cloud-First} structure: Optimal policy always assigns jobs to SM1 when the cloud is idle ($n_2 = i_1 = 0$), verified consistently across all system configurations (a)--(d).}
\label{fig:cloud_first}
\end{figure}

In Figure~\ref{fig:system_policies}, we display the policy over states in $\Sp \setminus \Sp_1$. States where no jobs are available in the base queue (i.e. $n_0 = 0$) are excluded, as \texttt{idle} is trivially the only admissible action. The plots, shown for the four system configurations (a)-(d), highlight the \textit{switch-type} structure of the SM2 assignments: the optimal policy exhibits a threshold behavior along both the base and cloud queue dimensions. Together, the figures confirm that the optimal policy maintains the \textit{cloud-first} and \textit{switch-type} structures across a range of system conditions.

We observe that the threshold for assigning SM2 increases with higher values of $f$ and $K$. This aligns with the intuition developed in Section~\ref{sec:main_result}: both larger $f$ and higher $K$ raise the expected total processing time under SM2 (i.e. $1/\mu_{l2} + 1/\mu_{c2}$), making SM1 more favorable and increasing the switching threshold to SM2.
Conversely, when system utilization $\rho$ increases, the threshold for assigning SM2 decreases. This is expected, as higher utilization necessitates offloading more jobs via SM2 for load balancing and throughput stability.


\begin{figure}[!t]
    \centering
    \begin{minipage}{0.47\linewidth}
        \centering
        \begin{minipage}{0.47\linewidth}
            \includegraphics[width=\linewidth]{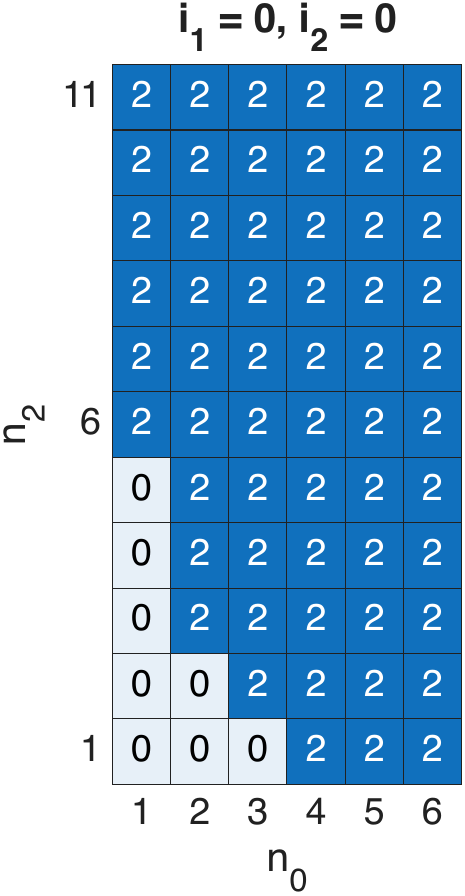}
        \end{minipage}
        \begin{minipage}{0.47\linewidth}
            \includegraphics[width=\linewidth]{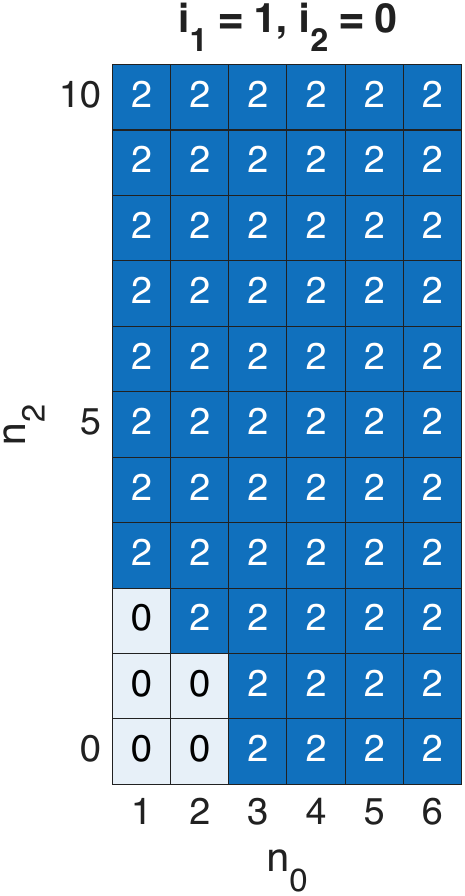}
        \end{minipage}
        \smallskip
        \small (a) Baseline: $\rho = 0.4,$ \\$f = 0.4, K = 8$
        	\vspace{2em}
    \end{minipage}
    \hfill
    \begin{minipage}{0.47\linewidth}
        \centering
        \begin{minipage}{0.47\linewidth}
            \includegraphics[width=\linewidth]{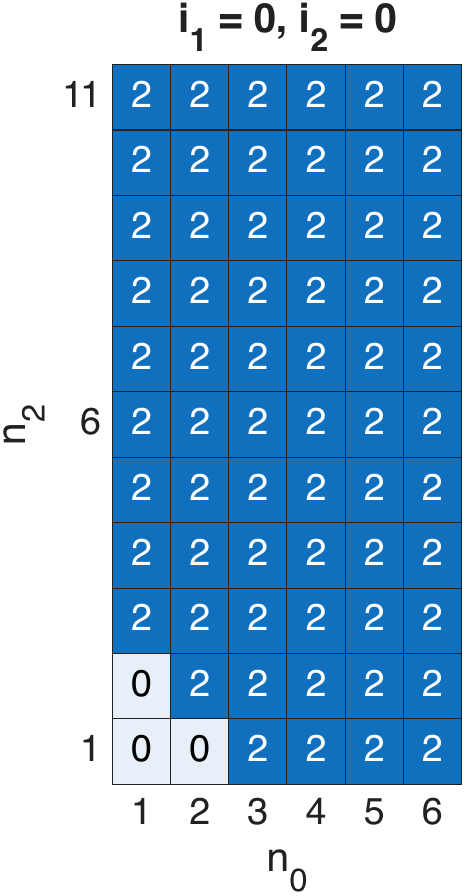}
        \end{minipage}
        \begin{minipage}{0.47\linewidth}
            \includegraphics[width=\linewidth]{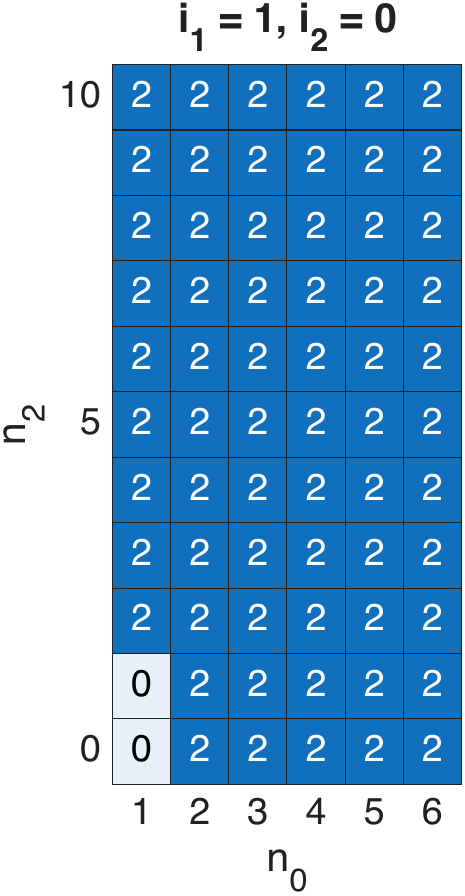}
        \end{minipage}
        \smallskip
        \small (b) High Utilization: $\rho = 0.8,$\\ $f = 0.4, K = 8$
        	\vspace{2em}
    \end{minipage}
	\begin{minipage}{0.47\linewidth}
        \centering
        \begin{minipage}{0.47\linewidth}
            \includegraphics[width=\linewidth]{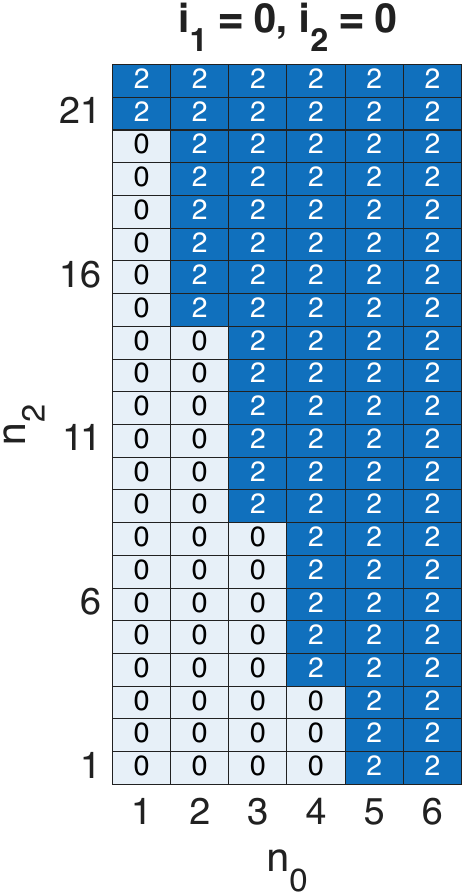}
        \end{minipage}
        \begin{minipage}{0.47\linewidth}
            \includegraphics[width=\linewidth]{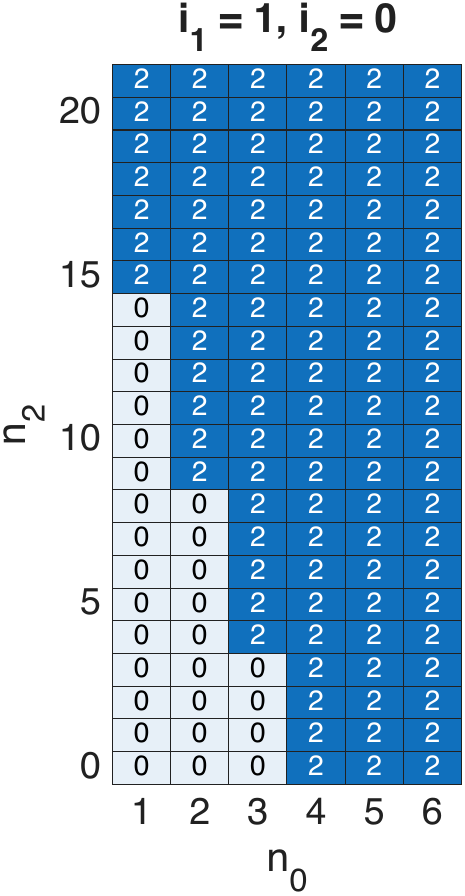}
        \end{minipage}
        \smallskip
        \small (c) More local SM2 processing: $\rho = 0.4, f = 0.8, K = 8$
       	\vspace{1em}
    \end{minipage}
    \hfill
    \begin{minipage}{0.47\linewidth}
        \centering
        \begin{minipage}{0.47\linewidth}
            \includegraphics[width=\linewidth]{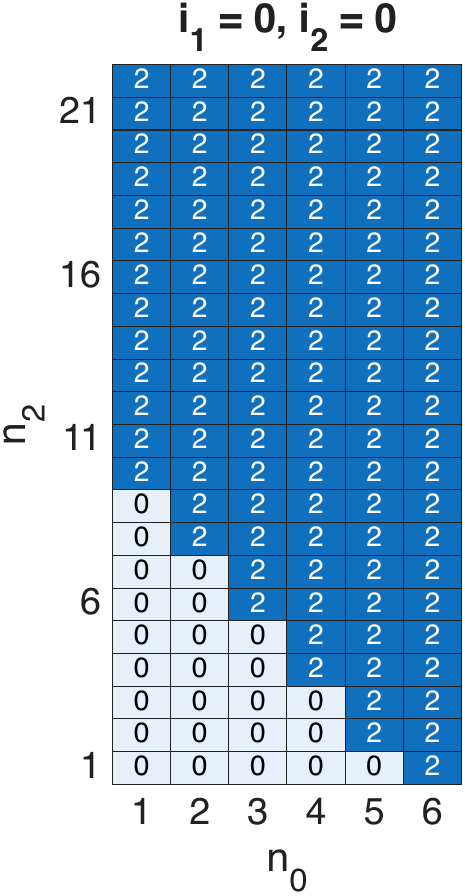}
        \end{minipage}
        \begin{minipage}{0.47\linewidth}
            \includegraphics[width=\linewidth]{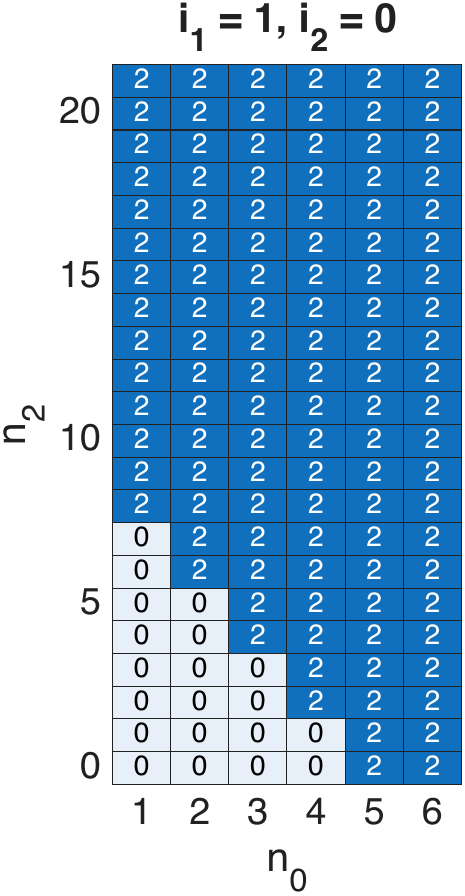}
        \end{minipage}
        \smallskip
        \small (d) Higher cloud capacity: \\ $\rho = 0.4, f = 0.4, K = 15$
        	\vspace{1em}
    \end{minipage}
    \caption{\textit{Switch-Type} structure: The optimal policy assigns jobs to SM2 in directions where queue lengths increase, beyond configuration-specific switching points, and idles in the opposite direction. Switching-points vary across configurations (a)--(d)}
    \label{fig:system_policies}
\end{figure}

\subsection{Delay Performance}
We compare the delay performance of the optimal \textit{switch-type} policy against two baseline policies:

(a) \textit{Offload-Only}, and  
(b) \textit{Non-Idling}.

The \textit{Offload-Only} policy exclusively assigns jobs to SM1, offloading all computation to the cloud while idling the local server. In contrast, the \textit{Non-Idling} policy adheres to the \textit{cloud-first} structure but aggressively utilizes SM2, assigning it whenever it is admissible.

Both baselines are \textit{cloud-first} and \textit{switch-type} but represent opposite extremes in SM2 usage. The \textit{Offload-Only} policy corresponds to the limiting case where switching thresholds are effectively infinite - SM2 is never used. On the other hand, the \textit{Non-Idling} policy represents the case where switching thresholds are zero - SM2 is always used when possible.

It is important to note that the \textit{Offload-Only} policy does not utilize local compute resources and is therefore not stabilizing at high arrival rates (specifically, for $\rho \geq K/(K+1)$). The \textit{Non-Idling} policy, however, remains stabilizing across all arrival rates.

We simulate a system with $f = 0.6$ and $K = 10$ and measure the delay performance across a range of system utilization levels, $\rho$. The results are shown in Figure~\ref{fig:performance}.

As expected, the optimal policy consistently outperforms both baselines across all values of $\rho$. At low arrival rates, the \textit{Offload-Only} policy performs close to optimal, while at high arrival rates, the \textit{Non-Idling} policy performs comparably. This trend is consistent with our earlier observation: increasing system utilization lowers the optimal switching thresholds for SM2, making its use more favorable at higher utilization levels.

\begin{figure}[!t]
\centering
\includegraphics[width=0.49\textwidth]{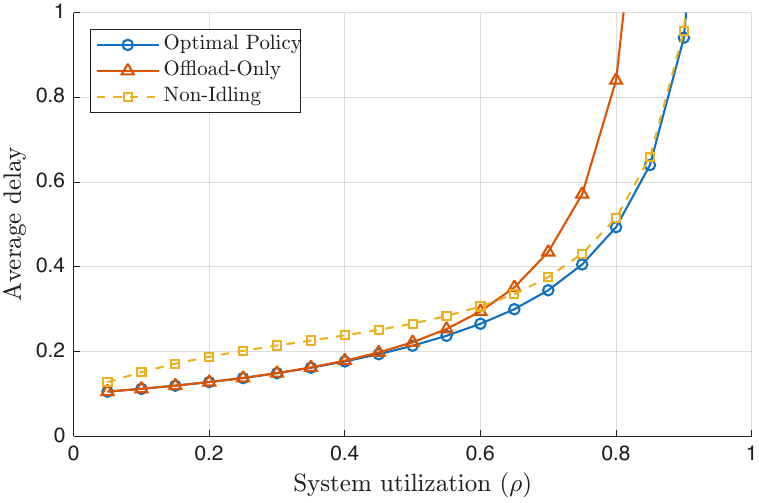}
\caption{Average delay under the optimal policy compared to two baseline policies --- \textit{Offload-Only} and \textit{Non-Idling}, as a function of system utilization $\rho$, for a fixed system with $f = 0.6$ and $K = 10$.}
\label{fig:performance}
\end{figure}

\section{Conclusion}\label{sec:conclusion}
We studied the problem of service mode assignment in a two-mode offloading system with the goal of minimizing average delay. In our model, service mode 1 fully offloads jobs to the cloud, while service mode 2 performs a combination of local and cloud processing, with the cloud assumed to have greater capacity than local resources.

We establish that when the cloud is idle, the optimal action is to immediately offload an available job. When the cloud is busy, we show under mild conditions that the optimal policy exhibits a \textit{switch-type} structure: assign a job to mixed service if the queue lengths exceed a threshold, and idle otherwise.

We verified this structural property through simulation and compared the delay performance of the optimal \textit{switch-type} policy against two baseline \textit{cloud-first} policies. 

Future directions include establishing a formal proof of the urgency-monotonicity conjecture and extending the analysis to more general multi-stage, multi-mode offloading systems.
\bibliographystyle{IEEEtran}
\bibliography{references.bib}
\appendices
\section{Proofs supporting Theorem~\ref{th:cloud_first}}
\subsection{Proofs supporting Proposition~\ref{pro:offload_empty}} 
\label{app:offload_empty}
\subsubsection{Proof of Lemma~\ref{lem:no_idle_empty}}\label{app:no_idle_empty}
Let $\pi$ be a delay-optimal policy with $\pi(s_0) = \mathtt{idle}$, where $s_0 = (n_0,0,0,0)$ for some $n_0\in\N$.
Consider a system $S$ with state process $S(t)$ starting in the state $s_0$.
Since the policy $\pi$ is stabilizing, the DMS cannot idle indefinitely while the base queue grows. Let $\tau$ be the first instant at which the DMS takes a non-idling action. We analyze this first action to show a contradiction.

\noindent\textit{Case 1:} $\pi(S(\tau)) = \mathtt{SM1}$\\ Construct a coupled system $\tilde S$ with state process $\tilde S(t)$ under policy $\tilde \pi$, which selects \texttt{SM1} at $t=0$ and replicates $\pi$ for $t>\tau$.
\[\tilde \pi(t) = \begin{cases}
    \texttt{SM1} & \text{at } t = 0 \\
    \texttt{idle} & \text{for } t \in (0,\tau]\\
    \pi(t) & \text{for } t > \tau
    \end{cases}\]
Due to the coupling, $\tilde \pi$ can infer $S(t)$ from $\tilde S(t)$ and hence $\tilde\pi$ is non-anticipatory. Since the job assigned to SM1 leaves sooner in $\tilde S$, the expected cost is strictly lower, contradicting the optimality of $\pi$.

\noindent\textit{Case 2:} $\pi(S(\tau)) = \mathtt{SM2}$\\ This case follows the same reasoning. Assigning at a job SM2 earlier reduces the expected cost.

Hence, initially idling is always suboptimal.\proofover


\subsubsection{Proof of Lemma~\ref{lem:no_SM2_empty}}\label{app:no_SM2_empty}
Let $\pi$ be a delay-optimal policy such that $\pi(s_0) = \mathtt{SM2}$, where $s_0 = (n_0,0,0,0)$ for some $n_0\in\N$.
Consider a system $S$ with state process $S(t)$ starting in the state $s_0$.
Define $\tau$ as the first time the DMS selects action \texttt{SM1} under policy $\pi$ and let $J_\beta(s;\pi, E)$ denote the expected discounted cost under policy $\pi$ from state $s$ conditional on event $E$:
\[J_\beta(s;\pi, E) \coloneqq \E_{\pi}\lsb \int_0^\infty e^{-\beta t} n(t)dt \mb S(0) = s, E \rsb\]
Now, construct a coupled system $\tilde S$ with state process $\tilde S(t)$ operating under a modified policy $\tilde \pi$, where $\tilde \pi$ assigns control \texttt{SM1} at $t = 0$.

Let $J$ denote the head-of-line job in the base queue of $S(0)$, and let $\tilde J$ denote the corresponding job in $\tilde S(0)$. By the coupling framework in Section~\ref{sec:coupling_framework}, $J$ and $\tilde J$ are associated with identical service time triplets $(\sigma_{c1}, \sigma_{l2}, \sigma_{c2})$.

Because $\tilde J$ is immediately assigned to the SM1 cloud in $\tilde S$, and $J$ is initially sent to the SM2 local server in $S$, the coupling relation implies:
\[\sigma_{l2} = \frac{\mu_{c1}}{\mu_{l2}} \:\sigma_{c1} \quad \Rightarrow \quad \sigma_{c1} < \sigma_{l2}.\]
Thus, the SM1 cloud service time for job $\tilde J$ in system $\tilde S$ is strictly shorter than the local SM2 service time of job $J$ in system $S$.

We analyze trajectories up to $t \leq \min\{\tau, \sigma_{c1}\}$. We have:
\[\tilde S(t) = S(t) + (0, -1, 1, 0), \quad \text{for } t \leq \min\{\tau, \sigma_{c1}\}.\]
Thus, both systems incur the same cost for $t \leq \min{\tau, \sigma_{c1}}$. We now distinguish two cases:
\smallskip
\noindent\textit{Case 1:} $\tau \leq \sigma_{c1}$.\\
At time $\tau$, $\pi$ assigns a job to SM1. In this case, let $\tilde \pi$ assign the corresponding coupled job to SM2 and follow an optimal policy thereafter. Since the state processes are in the same state and both follow the same policy, they incur the same expected cost for time $t \geq \sigma_{c1}$. Hence, both systems incur the same expected cost in this case, i.e.
\[J_\beta(s_0,\pi,\tau \leq \sigma_{c1}) = J_\beta(s_0,\tilde\pi,\tau \leq \sigma_{c1})\]

\smallskip
\noindent\textit{Case 2:} $\sigma_{c1} < \tau$.\\
In this case, $\tilde J$ completes cloud service before $\pi$ switches to SM1. At time $t = \sigma_{c1}$, both servers are idle in $\tilde S$ (since only one job has been processed), while in $S$, the job $J$ is still undergoing SM2 local service with residual time:
\[\sigma_{l2} - \sigma_{c1} = \sigma_{c1} \left(\frac{\mu_{l2}}{\mu_{c1}} - 1\right) > 0.\]
Since $\tilde \pi$ can infer the remaining processing time of $J$ from the coupling, it assigns a dummy job to the SM2 queue with the same remaining processing times as $J$, to couple the systems for $t > \sigma_{c1}$.
Since the dummy job does not incur any additional cost in $\tilde S$, we have
\[J_\beta(s_0,\pi,\tau > \sigma_{c1}) > J_\beta(s_0,\tilde\pi,\tau > \sigma_{c1})\]

\medskip
\noindent Finally, since the event $\{\tau > \sigma_{c1}\}$ occurs with positive probability, we have:
{\footnotesize
\begin{align*}
J_\beta(s_0,\pi) &- J_\beta(s_0,\tilde\pi)\\ 
&= \:\Prob(\tau \leq \sigma_{c1})\lsb J_\beta(s_0,\pi,\tau \leq \sigma_{c1}) - J_\beta(s_0,\tilde\pi,\tau \leq \sigma_{c1})\rsb \\
&\quad + \Prob(\tau > \sigma_{c1})\lsb J_\beta(s_0,\pi,\tau > \sigma_{c1}) - J_\beta(s_0,\tilde\pi,\tau > \sigma_{c1})\rsb	\\
&> 0
\end{align*}}
This contradicts the optimality of $\pi$. \proofover


\subsection{Proofs supporting Proposition~\ref{pro:empty_behaviour}} \label{app:empty_behaviour}
\subsubsection{Proof for Lemma~\ref{lem:empty_interval}} \label{app:empty_interval}
Let $\pi$ be optimal. Suppose for contradiction that there exist $m_1,m_2\in\N_0$ with $m_1<m_2$ such that:
\[\pi(m_1,0,1,0) = \pi(m_2,0,1,0) = \mathtt{SM2}\]
but for all intermediate states,
\[\pi(n,0,1,0) = \mathtt{idle} \quad\forall n\in {\llbracket m_1+1, m_2-1 \rrbracket}\]

Consider a system operating under $\pi$ starting at a state $(n_0,0,1,0)$ with $n_0\in\llbracket m_1+1,m_2-1\rrbracket$. This system will exclusively assign SM1 until the system either hits state $(m_1,0,1,0)$ right after an SM1 departure or state $(m_2,0,1,0)$ right after an arrival. In both cases, at this time, the system assigns a job to SM2 --  denote this job by $J$. Note however that $J$ was already present in the system at $t=0$. Similar to the argument in Proposition~\ref{pro:offload_empty}, assigning job $J$ to SM2 immediately at $t=0$, rather than waiting for some positive random amount of time, would strictly reduce the waiting cost, contradicting the optimality of $\pi$. \proofover


\subsubsection{Proof for Lemma~\ref{lem:par_comp_win}} \label{app:par_comp_win}
Let $\pi$ be optimal, satisfying
\[\pi(n,0,1,0) =\mathtt{idle} \quad \forall n\geq M_0\]
for some $M_0\in\N$. Consider a system operating under $\pi$ with state process $S(t)$ starting at $s_0 =(n_0,0,1,0)$ where $n_0 \geq M_0$\\
Define a modified policy, $\tilde\pi$, governing a system with state process $\tilde S(t)$ with initial state $\tilde s_0=(n_0,0,0,1)$. Let the cloud job present initially be denoted by $J$ and $\tilde J$ in states $s_0$ and $\tilde s_0$, respectively. \\

Couple the processes to have identical arrival processes and job sizes using the framework in Section~\ref{sec:coupling_framework}, except for $J$ and $\tilde J$. Let the service times for these jobs be $\sigma_{c1}$ and $\sigma_{c2}$ respectively. Couple them such that: $\sigma_{c2} = \frac{\mu_{c1}}{\mu_{c2}}\sigma_{c1} = (1-f)\,\sigma_{c1}$. This coupling is valid $\sigma_{c1}, \sigma_{c2}$ are exponentially distributed.

Define policy $\tilde \pi$ such that:
\[
\tilde S(t) = 
\begin{cases} 
	S(t) + (0,0,-1,+1), & \text{for } t \in [0,\sigma_{c2}], \\
	S(t) + (0,0,-1,0), & \text{for } t \in [\sigma_{c2}, \sigma_{c1}], \\
	S(t), & \text{for } t>\sigma_{c1}, \\
\end{cases}
\]
Here, $\tilde\pi$ mirrors $\pi$'s decisions for all jobs except $J$, serving $\tilde J$ when $\pi$ serves $J$. This definition of $\tilde \pi$ is valid and non-anticipative, as the service-time coupling and action mirroring ensures the $\tilde\pi$ can track $S(t)$.
We evaluate the expected cost reduction:
\begin{align*}
J_\beta(s_0;\pi) - J_\beta(\tilde s_0;\tilde\pi) 
&= \mathbb{E}\lsb\int_{\sigma_{c2}}^{\sigma_{c1}} e^{-\beta t} dt\rsb\\
&\geq \mathbb{E}\lsb e^{-\beta \sigma_{c1}} (\sigma_{c1} - \sigma_{c2}) \right]\\
&\geq f \mathbb{E} \left[e^{-\beta \sigma_{c1}} \sigma_{c1} \rsb
\end{align*}

To show that this cost reduction is bounded away from zero by a constant $\epsilon_\beta>0$ independent of $n_0$, define the event $E$ as the event where job $J$ completes service before any other state transition in $S$. Then, $\Prob(E) = \frac{\mu_{c1}}{\lambda+\mu_{c1}}$ and conditioned on $E$, $\sigma_{c1}$ is exponential with rate $\lambda + \mu_{c1}$. Hence, we have:
\begin{align*}
J_\beta(s_0;\pi) - J_\beta(\tilde s_0;\tilde\pi)
&\geq f\frac{\mu_{c1}}{\lambda+\mu_{c1}} \E\lsb e^{-\beta \sigma_{c1}} \sigma_{c1}\mb E \rsb \\
&\coloneqq \epsilon_\beta 
\end{align*}
Finally, by optimality of $\pi$, it follows that:
\[V_\beta(s_0) - V_\beta(\tilde s_0) \geq V_\beta(s_0) - J_\beta(\tilde s_0;\tilde\pi) \geq \epsilon_\beta\]
\proofover

\subsubsection{Proof for Lemma~\ref{lem:eventually_SM2}} \label{app:eventually_SM2}
Suppose for contradiction that $\pi$ is delay optimal policy and satisfies
\[\pi(n,0,1,0) = \mathtt{idle} \quad \forall n\geq M_0\]
for some $M_0\in\N$.
Consider a system under $\pi$ with state process $S(t)$ with initial state $s_m = (M_0+m,0,1,0)$, for some $m\in\N$. Define $\tau$ as the first time this system reaches state $(M_0,0,1,0)$. Because $\pi$ is stabilizing and optimal, we have $\tau<\infty$ w.p. 1.

Define a modified policy, $\tilde \pi$ with the same initial state that assigns SM2 at $t=0$, resulting in state process $\tilde S(t)$. Couple the state processes using the framework in Section~\ref{sec:coupling_framework}. Let $\sigma_{l2}$ denote the local service time for the job assigned SM2 at $t=0$ in $\tilde S$ and let $\tilde \pi$ be such that
\[\tilde S(t) = U_2[S(t)], \quad \text{ for } t \in [0,\min\{\tau,\sigma_{l2}\})\]

After time $\min\{\tau, \sigma_{l2}\}$, $\tilde{\pi}$ simply operates using the optimal policy, $\pi$.

Define \[J_\beta(s;\pi, E) \coloneqq \E_{\pi}\lsb \int_0^\infty e^{-\beta t} n(t)dt \mb S(0) = s, E \rsb\]
Then,
\begin{align*}
J_\beta&(s_m;\pi,\sigma_{l2}<\tau)-J_\beta(s_m;\tilde\pi,\sigma_{l2}<\tau)\\ 
&=
\E \lsb e^{-\beta \sigma_{l2}}\rsb\lsb V_\beta(S(\sigma_{l2})) - V_\beta(\tilde S(\sigma_{l2}))\rsb\\
&\geq e^{-\beta\mathbb E[\sigma_{l2}]}\epsilon_\beta
= e^{-\beta/\mu_{l2}} \epsilon_\beta\\
&\coloneqq \epsilon_\beta^1 
\end{align*}
Where the inequality follows using Lemma~\ref{lem:par_comp_win}. Similarly,
\begin{align*}
J_\beta&(s_m;\tilde\pi,\tau\leq\sigma_{l2})-J_\beta(s_m;\pi,\tau\leq\sigma_{l2}) \\
&=\E \lsb e^{-\beta \tau}\rsb\lsb V_\beta(M_0-1,1,1,0) - V_\beta(M_0,0,1,0)\rsb\\
&\leq \lsb V_\beta(M_0-1,1,1,0) - V_\beta(M_0,0,1,0)\rsb\\
&\coloneqq \epsilon_\beta^2
\end{align*}
Finally, note that
As $m \rightarrow \infty$, we have $\Prob (\tau < \sigma_{l2}) \rightarrow 0$ and $\Prob (\tau \geq \sigma_{l2}) \rightarrow 1$. Thus, for sufficiently large $m$.
\[J_\beta(s_m;\pi) - J_\beta(s_m;\tilde\pi) \geq \mathbb \Prob(\sigma_{l2}<\tau) \epsilon_\beta^1 - \Prob(\sigma_{l2}>\tau) \epsilon_\beta^2 > 0\]
contradicting the optimality of $\pi$.\proofover


\subsection{Proofs supporting Proposition~\ref{pro:SM1_induction}} \label{app:SM1_induciton}

\subsubsection{Proof for Lemma~\ref{lem:idle_SM2}} \label{app:idle_SM2}
We prove by contradiction. Suppose $\pi$ is optimal and satisfies:
\begin{align*}
	\pi(s_0) = \mathtt{idle},\quad
	\pi(As_0) = \mathtt{SM1},\quad
	\pi(D_Ls_0) = \mathtt{SM2}
\end{align*}
where, $s_0 = (n_0, 1, 0, 0)$ for some $n_0\in\mathbb N$.\\
Let $S(t)$ denote the state process for a system operating under $\pi$ with initial state $s_0$. Using the coupling framework described in Section~\ref{sec:coupling_framework}, consider a coupled system with state process $\tilde S(t)$, starting at the same initial state, under a modified policy, $\tilde \pi$ such that $\tilde \pi$ picks DMS control \texttt{SM1} at $t=0$ and thereafter uses the optimal policy, $\pi$.

Let $J, J'$ denote the job in the SM2 cloud queue and the HOL job from the base queue in $S(0)$. Additionally, Similarly, let $\tilde{J}$ and $\tilde{J}'$ denote the SM1 and SM2 jobs in $\tilde{S}(0^+)\;(= U_1 s_0)$ respectively. Couple the job sizes in $S,\tilde S$ by mapping $J$ to $\tilde J$ and $J'$ to $\tilde J'$ and all other jobs in $S$ to the corresponding jobs in $\tilde S$.
\begin{remark}
	We remind the reader that we have coupled the SM2 local job in $S$ to the SM1 cloud job in $\tilde S$ and \textit{not} the SM2 local job in $\tilde S$.
\end{remark}
We use notation $\sigma_{l2}^J \;(\sim\text{Exp}(\mu_{l2}))$ to denote the SM2-local service time for job $J$ and $\tau_a$ to denote the time of the first arrival in $S$ (also $\tilde S$ from coupling).
Define \[J_\beta(s;\pi, E) \coloneqq \E_{\pi}\lsb \int_0^\infty e^{-\beta t} n(t)dt \mb S(0) = s, E \rsb\]

Now, the time of the first transition instant in $\tilde S$ is given by
\[\tau \coloneqq \min\{\tau_a, \sigma_{l2}^{\tilde J'}, \sigma_{c1}^{\tilde J}\}\]
We proceed by considering cases
\begin{itemize}[leftmargin=1em]
	\item \textbf{Case 1:} $\tau = \tau_a$ \\
	This corresponds to $\tau_a$ being the first decision instant on both systems.
	Since $\pi(As_0) = \mathtt{SM1}$, we have $S(\tau^+) = U_1As_0 = \tilde S(\tau^+)$.
	In this case, the systems have the same expected cost, i.e.
	\[J_\beta(s_0;\pi,\tau=\tau_a) = J_\beta(s_0;\tilde\pi,\tau=\tau_a)\]
		
	\item \textbf{Case 2:} $\tau = \sigma_{c1}^{\tilde J}$ \\
	Now, $\tilde S(\tau) = D_1U_1s_0$ while $S(\tau) = s_0$ with $J$'s remaining service time being $\sigma_{c1}^{\tilde J} \lp\frac{\mu_{c1}}{\mu_{l2}}-1\rp$. As previously argued by comparing systems by introducing a dummy job (cf Lemma~\ref{lem:no_SM2_empty}), the cost would be strictly improved if we consider the state where $J$ completes local service.

	Since both systems incur the same cost for $t<\tau$, we take the difference to get: 
	\begin{align*}
		J_\beta&(s_0;\pi,\tau=\sigma_{c1}^{\tilde J}) - J_\beta(s_0;\tilde\pi,\tau =\sigma_{c1}^{\tilde J})\\
		&> \E[e^{-\beta\tau}]\lsb V_\beta(D_Ls_0;\pi) - V_\beta(D_1U_1s_0;\pi) \rsb
	\end{align*}
	Since $\pi(D_Ls_0) = \mathtt{SM2}$, 
	\[V_\beta(D_Ls_0;\pi) = V_\beta(U_2D_Ls_0;\pi)\]
	Now, a system in state $D_1U_1s_0$ can track a system in state $U_2D_Ls_0$ by introducing a dummy job as we argued previously (cf Lemma~\ref{lem:no_SM2_empty}). Hence:
	\[V_\beta(U_2D_Ls_0) > V_\beta(D_1U_1s_0)\]
	Finally, combining these results we get:
	\[J_\beta(s_0;\pi,\tau=\sigma_{c1}^{\tilde J}) > J_\beta(s_0;\tilde\pi,\tau=\sigma_{c1}^{\tilde J})\]
	
	\item \textbf{Case 3:} $\tau = \sigma_{l2}^{\tilde J'}$ \\
	Now, $\tilde S(\tau) = s_0$ while $S(\tau) = D_LU_1s_0$. Once again, by comparing systems by introducing a dummy job (cf Lemma~\ref{lem:no_SM2_empty}), we can show
	\[V_\beta(U_1s_0) \geq V_\beta(D_LU_1s_0)\]
	Using this, we have
	\begin{align*}
		J_\beta&(s_0;\pi,\tau=\sigma_{c1}^{\tilde J}) - J_\beta(s_0;\tilde\pi,\tau=\sigma_{c1}^{\tilde J})\\
		&= \E[e^{-\beta\tau}]\lsb V_\beta(s_0) - V_\beta(D_LU_1s_0) \rsb\\
		&\geq \E[e^{-\beta\tau}]\lsb V_\beta(s_0) - V_\beta(U_1s_0) \rsb\\
		&= \E[e^{-\beta\tau}]\lsb J_\beta(s_0;\pi) - J_\beta(s_0;\tilde\pi) \rsb
	\end{align*} 
\end{itemize}
\noindent
Finally,
\begin{align*}
	J_\beta&(s_0;\pi) - J_\beta(s_0;\tilde\pi)\\
	=&\Prob(\tau = \tau_a)\lp J_\beta(s_0;\pi,\tau=\tau_a) - J_\beta(s_0;\tilde\pi,\tau=\tau_a) \rp + \\
	& \Prob(\tau = \sigma_{c1}^{\tilde J})\lp J_\beta(s_0;\pi,\tau=\sigma_{c1}^{\tilde J}) - J_\beta(s_0;\tilde\pi,\tau=\sigma_{c1}^{\tilde J}) \rp + \\
	&\Prob(\tau = \sigma_{l2}^{\tilde J'})\lp J_\beta(s_0;\pi,\tau=\sigma_{c1}^{\tilde J}) - J_\beta(s_0;\tilde\pi,\tau=\sigma_{c1}^{\tilde J}) \rp
\end{align*}
Substituting the results shown above, we have
\[J_\beta(s_0;\pi) > J_\beta(s_0;\tilde\pi)\]
which contradicts the optimality of $\pi$. 

\subsubsection{Proof for Lemma~\ref{lem:idle_idle}} \label{app:idle_idle}
We prove by contradiction. Suppose $\pi$ is optimal and satisfies:
\[\pi(s_0) = \pi(D_Ls_0) = \mathtt{idle}, \:\: \pi(As_0) = \pi(A^2s_0) = \mathtt{SM1}\]
where, $s_0 = (n_0, 1, 0, 0)$ for some $n_0\in\N$.\\
Let $S(t)$ denote the state process for a system operating under $\pi$ with initial state $s_0$. Using the coupling framework described in Section~\ref{sec:coupling_framework}, consider a coupled system with state process $\tilde S(t)$, starting at the same initial state, under a modified policy, which picks DMS control \texttt{SM1} at $t=0$ and thereafter follows the optimal policy, $\pi$.

Let $J, J'$ denote the job in the SM2 cloud queue and the HOL job from the base queue in $S(0)$. Additionally, Similarly, let $\tilde{J}'$ and $\tilde{J}$ denote the SM1 and SM2 jobs in $\tilde{S}(0^+)\;(= U_1 s_0)$ respectively. Couple the job sizes in $S,\tilde S$ by mapping $J$ to $\tilde J$ and $J'$ to $\tilde J'$ and all other jobs in $S$ to the corresponding jobs in $\tilde S$.
\begin{remark}
	We remind the reader that we have coupled the SM2 local job in $S$ to the SM2 local job in $\tilde S$ and \textit{not} the SM2 local job in $\tilde S$ unlike in the previous proof.
\end{remark}

Similar to the previous proof, we use notation $\sigma_{l2}^J \;(\sim\text{Exp}(\mu_{l2}))$ to denote the SM2-local service time for job $J$ and $\tau_a$ to denote the time of the first arrival in $S$ (also $\tilde S$ from coupling).
Define \[J_\beta(s;\pi, E) \coloneqq \E_{\pi}\lsb \int_0^\infty e^{-\beta t} n(t)dt \mb S(0) = s, E \rsb\]

Now, the time of the first decision instant in $\tilde S$ is given by
\[\tau \coloneqq \min\{\tau_a, \sigma_{l2}^{\tilde J'}, \sigma_{c1}^{\tilde J}\}\]
Under our coupling, we have $n(t) = \tilde n(t)$ for $t<\tau$. Equivalently, both systems incur the same delay cost for time $t<\tau$. Additionally, since $\tilde\pi$ uses the optimal policy after the first decision instant the difference in the cost-to-go under $\pi$ and $\tilde\pi$ will be equal to the expected (discounted) difference in value function evaluated at the states of the systems at the time of the first transition. Mathematically,
\begin{align}
&J_\beta(s_0;\pi) - J_\beta(s_0;\tilde\pi)\nonumber\\
& =\Prob(\tau = \tau_a)\lsb J_\beta(s_0;\pi,\tau = \tau_a) - J_\beta(s_0;\tilde\pi,\tau = \tau_a) \rsb \nonumber\\
&\quad + \Prob(\tau = \sigma_{l2}^{\tilde J'})\lsb J_\beta(s_0;\pi,\tau = \sigma_{l2}^{\tilde J'}) - J_\beta(s_0;\tilde\pi,\tau = \sigma_{l2}^{\tilde J'}) \rsb \nonumber\\
&\quad + \Prob(\tau = \sigma_{c1}^{\tilde J})\lsb J_\beta(s_0;\pi,\tau = \sigma_{c1}^{\tilde J}) - J_\beta(s_0;\tilde\pi,\tau = \sigma_{c1}^{\tilde J}) \rsb \nonumber\\
& =\Prob(\tau = \tau_a)\E\lsb e^{-\beta\tau}\rsb\lsb V_\beta(As_0) - V_\beta(AU_1s_0) \rsb \nonumber\\
&\quad + \Prob(\tau = \sigma_{l2}^{\tilde J'})\E\lsb e^{-\beta\tau}\rsb\lsb V_\beta(D_Ls_0) - V_\beta(D_LU_1s_0) \rsb \nonumber\\
&\quad + \Prob(\tau = \sigma_{c1}^{\tilde J})\E\lsb e^{-\beta\tau}\rsb\lsb V_\beta(s_0) - V_\beta(D_1U_1s_0) \rsb \nonumber\\
& = \frac{\lambda\lsb V_\beta(As_0) - V_\beta(AU_1s_0) \rsb}{\lambda + \mu_{l2} + \mu_{c1} + \beta} \nonumber\\
&\quad + \frac{\sigma_{l2}\lsb V_\beta(D_Ls_0) - V_\beta(D_LU_1s_0)\rsb}{\lambda + \mu_{l2} + \mu_{c1} + \beta}\nonumber\\
&\quad + \frac{\mu_{c1}\lsb V_\beta(s_0) - V_\beta(D_1U_1s_0) \rsb}{\lambda + \mu_{l2} + \mu_{c1} + \beta}
 \label{eqn:pro4.6_rhs}
\end{align}

We now look at each term from the RHS in \eqref{eqn:pro4.6_rhs} under the following cases

\begin{itemize}[leftmargin=1em]
	\item \textbf{Term 1:} $V_\beta(As_0) - V_\beta(AU_1s_0)$ \\
	Since $\pi(As_0) = \mathtt{SM1}$, we have $V_\beta(AU_1s_0) = V_\beta(As_0)$ and the first term is zero.

	\item \textbf{Term 2:} $V_\beta(D_Ls_0) - V_\beta(D_LU_1s_0)$ \\
	Since $\pi(D_Ls_0) = \mathtt{idle}$, we have from Claim~\ref{claim:B1}
	\begin{align*}
	V_\beta(D_Ls_0) &- V_\beta(D_LU_1s_0) \\&> \dfrac{\lambda\lsb V_\beta(AD_Ls_0) - V_\beta(AD_LU_1s_0) \rsb}{\lambda + \mu_{c2} + \beta}
	\end{align*}

	\item \textbf{Term 3:} $V_\beta(s_0) - V_\beta(D_1U_1s_0)$ \\
	Since $AD_1U_1s_0$ is simply $s_0$, we have from Claim~\ref{claim:B2}
	\begin{align*}
	V_\beta(s_0) &- V_\beta(D_1U_1s_0) \\&> \dfrac{\lambda\lsb V_\beta(As_0) - V_\beta(s_0) \rsb}{\lambda + \mu_{c2} + \beta}
	\end{align*}
\end{itemize}
Substituting these in \eqref{eqn:pro4.6_rhs}, we get

\begin{align*}
&J_\beta(s_0;\pi) - J_\beta(s_0;\tilde\pi)\\
&>\frac{\lambda \mu_{l2}\lsb V_\beta(AD_Ls_0) - V_\beta(AD_LU_1s_0) \rsb}{\lambda + \mu_{l2} + \mu_{c1} + \beta}\\
&\quad + \frac{\lambda \mu_{c1}\lsb V_\beta(As_0) - V_\beta(s_0) \rsb}{\lambda + \mu_{l2} + \mu_{c1} + \beta}
\end{align*}
\noindent
Finally, since $\pi(As_0) = \pi(A^2s_0) = \mathtt{SM1}$, we have from Claim~\ref{claim:B3}
{\footnotesize
\[\mu_{l2} \lsb V_\beta(AD_Ls_0)-V_\beta(AD_LU_1s_0) \rsb + \mu_{c1} \lsb V_\beta(As_0)-V_\beta(s_0) \rsb > 0\]}
and Hence,
\[J_\beta(s_0;\pi) > J_\beta(s_0;\tilde\pi)\]
showing $\pi$ is not optimal. \proofover

\subsubsection{Supporting Claims}
\begin{claim}\label{claim:B1}
Let $\pi$ be a delay-optimal policy with 
\[\pi(s_0) = \mathtt{idle}\] 
where $s_0 = (n_0,0,0,1)$ for some $n_0\in\mathbb N$. Then,
\begin{align*}
V_\beta(s_0)-V_\beta(U_1s_0) \geq \frac{\lambda \lsb V_\beta(As_0) - V_\beta(AU_1s_0) \rsb}{\lambda + \mu_{l2} + \beta} \end{align*}
\end{claim}
\begin{proof}
	Consider a system $S$, operating under policy $\pi$ with state process $S(t)$ and initial state $s_0$.
	Using the coupling framework described in Section~\ref{sec:coupling_framework}, consider a coupled system, $\tilde S$, operating under a modified policy $\tilde\pi$, with state process $\tilde S(t)$ and initial state $U_1s_0$. Let $\tilde\pi$ pick DMS action \texttt{idle} till the first decision instant, after which it follows an optimal policy. Let $J,\tilde J$ denote the SM2 cloud jobs in $S(0), \tilde S(0)$ repectively. Couple the state processes to have the same arrival process and job sizes (mapping the SM1 cloud job in $\tilde S(0)$ to a job form the base queue). Let $\tau_a$ denote the time of the first arrival.\\
	The time of the first decision instant (for both systems) is given by
	\[\tau \coloneqq \min\{\sigma_{c2}^J, \tau_a\}\]
	Then $N(t) = \tilde N(t)$ for $t<\tau$.
	Defining 
	\[J_\beta(s;\pi, E) \coloneqq \E_{\pi}\lsb \int_0^\infty e^{-\beta t} n(t)dt \mb S(0) = s, E \rsb\]
	and proceeding similar to previous analysis, we have:
	{\footnotesize
	\begin{align}
	& J_\beta(s_0;\pi) - J_\beta(U_1s_0;\tilde\pi)\nonumber\\
	&= \E \lsb e^{-\beta\tau}\rsb 
	\big{\{} 
	\Prob(\tau = \tau_a) \lsb V_\beta(As_0) - V_\beta(AU_1s_0) \rsb \nonumber\\
	&\qquad\qquad\qquad\qquad+ \Prob(\tau = \sigma_{c2}^J) \lsb V_\beta(D_2s_0) - V_\beta(D_2U_1s_0) \rsb
	\big{\}} \nonumber\\
	&= \frac{\lambda \lsb V_\beta(As_0) - V_\beta(AU_1s_0) \rsb + \mu_{c2} \lsb V_\beta(D_2s_0) - V_\beta(D_2U_1s_0) \rsb}{\lambda + \mu_{l2} + \beta} \label{eqn:B1_rhs}
	\end{align}}
	From Proposition~\ref{pro:offload_empty}, we have $\pi(D_2s_0) = \mathtt{SM1}$. Hence,
	\[V_\beta(D_2s_0) = V_\beta(D_2U_1s_0)\]
	Removing this term from \eqref{eqn:B1_rhs}, we get
	{\footnotesize
	\begin{align*}
	V_\beta(s_0) - V_\beta(U_1s_0) &\geq J_\beta(s_0;\pi) - J_\beta(U_1s_0;\tilde\pi)\\
	&= \frac{\lambda \lsb V_\beta(As_0) - V_\beta(AU_1s_0) \rsb}{\lambda + \mu_{l2} + \beta} 
	\end{align*}}
	which is the desired result.
\end{proof}
\begin{claim}\label{claim:B2}
Let $\pi$ be a delay-optimal policy with
\[\pi(As_0) = \mathtt{idle}\]
where $s_0 = (n_0,1,0,0)$ for some $n_0\in\mathbb N$. Then
\begin{align*}
V_\beta(As_0)-V_\beta(s_0) > \frac{\lambda \lsb V_\beta(A^2s_0) - V_\beta(As_0) \rsb}{\lambda + \mu_{c2} + \beta} \end{align*}
\end{claim}

\begin{proof}
	This follows similarly to Claim \ref{claim:B1}, by considering a coupled systems $S,\tilde S$ starting from initial state $As_0, s_0$ respectively, where $S$ evolves under $\pi$ and $\tilde S$ under $\tilde \pi$ which picks DMS action $\mathtt{idle}$ at $t=0$ and uses an optimal policy after the first decision instant, we get:
	{\footnotesize
	\begin{align}
	& J_\beta(As_0;\pi) - J_\beta(s_0;\tilde\pi)\nonumber\\
	&= \frac{\lambda \lsb V_\beta(A^2s_0) - V_\beta(As_0) \rsb + \mu_{l2} \lsb V_\beta(D_LAs_0) - V_\beta(D_Ls_0) \rsb}{\lambda + \mu_{l2} + \beta} \label{eqn:B2_rhs}
	\end{align}}
	As previously argued by comparing systems by introducing a dummy job (cf Lemma~\ref{lem:no_SM2_empty}), we can show with strict inequality:
	{\footnotesize
	\begin{align}
	V_\beta(A^2s_0) > V_\beta(As_0) \quad \text{ and } \quad V_\beta(D_LAs_0) > V_\beta(D_Ls_0)
	\end{align}}
	By removing the second term from \eqref{eqn:B2_rhs} and multiplying the first term by $\frac{\lambda + \mu_{l2} + \beta}{\lambda + \mu_{c2} + \beta} < 1$, we get
	{\footnotesize
	\begin{align*}
	V_\beta(As_0) - V_\beta(s_0) &\geq J_\beta(As_0;\pi) - J_\beta(s_0;\tilde\pi)\\
	&> \frac{\lambda \lsb V_\beta(A^2s_0) - V_\beta(As_0) \rsb}{\lambda + \mu_{c2} + \beta} 
	\end{align*}}
	which is the desired result.
\end{proof}
\begin{claim}\label{claim:B3}
Let $\pi$ be a delay-optimal policy with
\[\pi(s_0) = \pi(As_0) = \mathtt{SM1}\]
where $s_0 = (n_0,1,0,0)$ for some $n_0\in\mathbb N$. Then
{\footnotesize
\[
\mu_{l2} \lsb V_\beta(D_Ls_0)-V_\beta(D_LU_1s_0) \rsb + \mu_{c1} \lsb V_\beta(s_0)-V_\beta(D_1U_1s_0) \rsb > 0 \]}
\end{claim}

\begin{proof}
	This follows similarly to Claim \ref{claim:B1}, by considering a coupled systems $S,\tilde S$ both starting from initial state $s_0$, where $S$ evolves under $\pi$ and $\tilde S$ under $\tilde \pi$ which picks DMS action $\mathtt{idle}$ at $t=0$ and uses an optimal policy after the first decision instant, we get:
	{\footnotesize
	\begin{align}
	& J_\beta(s_0;\tilde\pi) - J_\beta(s_0;\pi)\nonumber\\
	&= \frac{\lambda \lsb V_\beta(As_0) - V_\beta(AU_1s_0) \rsb + \mu_{l2} \lsb V_\beta(D_Ls_0) - V_\beta(D_LU_1s_0) \rsb}{\lambda + \mu_{l2} + \mu_{c1} + \beta}\nonumber\\
	&\qquad+ \frac{\mu_{c1} \lsb V_\beta(s_0) - V_\beta(D_1U_1s_0) \rsb}{\lambda + \mu_{l2} + \mu_{c1} + \beta}\nonumber\\
	&= \frac{\mu_{l2} \lsb V_\beta(D_Ls_0) - V_\beta(D_LU_1s_0) \rsb + \mu_{c1} \lsb V_\beta(s_0) - V_\beta(D_1U_1s_0) \rsb}{\lambda + \mu_{l2} + \mu_{c1} + \beta}\label{eqn:B3_rhs}
	\end{align}}
	where the second equality follows since we have $\pi(As_0) = \mathtt{SM1}$, and hence
	$V_\beta(As_0) = V_\beta(AU_1s_0)$.
	Since $\pi$ is optimal, the term on the RHS of \eqref{eqn:B3_rhs} must be positive.
\end{proof}

%
%

\section{Proofs supporting Theorem~\ref{th:partial_switch}} 
\subsection{Proof of Proposition~\ref{pro:wall_induction}} \label{app:wall_induction}
Let $\pi$ be a delay-optimal policy and \textit{cloud-first}. Suppose for some $s\in\Sp$, we have
\[\pi(s + (x_0,0,0,0)) =\mathtt{SM2} \quad \forall x_0\in\N\]
Suppose, contrary to the claim, that:
\[\pi(s + (0,0,0,1)) = \mathtt{idle}.\]
Consider a system with state process $S(t)$, initial state $s + (0,0,0,1)$, operating under policy $\pi$. Let job $J$ be the job at the tail of the base queue at $t=0$. Since the SM2 cloud queue has at least one job $t=0$ and the cloud scheduling is SM2-prioritizing, the cloud will exclusively serve SM2 jobs until this queue is empty.

Define $\sigma_{c2}$ as the time of the first SM2 cloud departure (with $\sigma_{c2} \sim \text{Exp}(\mu_{c2})$), and $\tau$ as the first time a job is assigned SM2. Since the SM2 local queue is empty at $t=0$, for $t \leq \tau$, the state evolves only via arrivals or SM2 cloud departures. Thus, we have $\tau < \sigma_{l2}$, because at the first SM2 cloud departure, the system transitions into a state $s + (x_0,0,0,0)$ for some $x_0 \in \mathbb N_0$, where the policy assigns SM2.

As the base queue only grows until $t=\tau$, we may assume job $J$ is assigned SM2 at time $\tau$. However, note that job $J$ was already in the system at $t=0$. As previously argued (cf. Proposition~\ref{pro:offload_empty}), the cost incurred by waiting before assigning job $J$ to SM2 could have been strictly reduced by immediately assigning SM2 at $t=0$ instead of waiting a positive, random time before taking this action anyway. Hence, policy $\pi$ is not optimal, contradicting our assumption.\proofover


\subsection{Proof of Proposition~\ref{pro:the_wall}} \label{app:the_wall}
From Lemma~\ref{lem:SM1_afterN_0}, we know that
\[\pi^*(s + (x_0, 0, 0, 0)) = \mathtt{SM2} \quad \forall x_0 \in \mathbb N_0\]
Applying Proposition~\ref{pro:wall_induction} gives us Part~\ref{pro:parta}. For Part~\ref{pro:partb} observe that for all states $(N_0+1, 0, 0, 0) + (x_0, 0, 0, 0)$ for $x_0\in\N_0$ the optimal policy assigns both SM1 and SM2 in succession. Since SM2 is always assigned, the argument used in the proof of Proposition~\ref{pro:wall_induction} extends to this case giving:
\[\pi^*(s' + (x_0, 0, 0, 0)) = \mathtt{SM2} \quad \forall x_0 \in \mathbb N_0\]
The result follows from applying Proposition~\ref{pro:wall_induction}.\proofover


\subsection{Proof of Lemma~\ref{lem:dec_switches}} \label{app:dec_switches}
By construction, for all $k \in \mathbb N$:
\[\pi^*(m, 0, 1, k) = \mathtt{SM2} \quad \forall m \geq N_k.\]
From Proposition~\ref{pro:the_wall},
\[\pi^*(m, 0, 0, k+1) = \mathtt{SM2} \quad \forall m \geq N_k,\]
implying the sequence is non-increasing:
\[N_0 \geq N_1 \geq N_2 \geq \dots\]
A similar argument holds for $\{N_k'\}_{k \in \N}$.\proofover

\end{document}